\documentclass[12pt]{article}
\usepackage{url}

\usepackage{amsfonts,amsthm,amsmath,amssymb,latexsym,comment}
\newcommand{\citeyear}{\cite}
\usepackage{epsfig}
\usepackage{graphicx}
\usepackage{enumitem}
\setlist[itemize,1]{label=$\bullet$}
\setlist[itemize,2]{label=$-$}

\newcommand{\commentout}[1]{}

\newcommand{\T}{\mathcal{T}}
\renewcommand{\S}{\mathcal{S}}

\newcommand{\ACT}{\mathit{CT}}
\newcommand{\shortv}[1]{#1}
\newcommand{\fullv}{\commentout}
\newcommand{\arxiv}[1]{#1}
\newcommand{\disc}{\commentout}
\newcommand{\dist}{\mathit{dist}}

\newtheorem{theorem}{Theorem}[section]
\newtheorem{proposition}[theorem]{Proposition}
\newtheorem{definition}[theorem]{Definition}
\newtheorem{corollary}[theorem]{Corollary}
\newtheorem{lemma}[theorem]{Lemma}
\arxiv{\usepackage{hyperref}}
\arxiv{\usepackage{fullpage}}
\newcommand{\wbox}{\mbox{$\sqcap$\llap{$\sqcup$}}}

\def\msgbound{\mbox{$2^{O(N\log(N))}$}}
\def\mifunction{f}

\title{Implementing Mediators with Asynchronous Cheap Talk}

\disc{
\author{Ittai Abraham}{VMware, Israel}{ittaia@cs.huji.ac.il}{}{}
\author{Danny Dolev}{The Hebrew University of Jerusalem\\ Jerusalem, Israel}{dolev@cs.huji.ac.il}{}{}
\author{Ivan Geffner}{Cornell University, Ithaca, NY}{ieg8@cornell.edu}{}{}
\author{Joseph Y. Halpern}{Cornell University, Ithaca, NY}{halpern@cs.cornell.edu}{}{}

\authorrunning{I. Abraham, D. Dolev, I. Geffner and J.\,Y. Halpern}

\Copyright{Ittai Abraham, Danny Dolev, Ivan Geffner, Joseph Y.
  Halpern}
  

\EventEditors{John Q. Open and Joan R. Access}
\EventNoEds{2}
\EventLongTitle{42nd Conference on Very Important Topics (CVIT 2016)}
\EventShortTitle{CVIT 2016}
\EventAcronym{CVIT}
\EventYear{2016}
\EventDate{December 24--27, 2016}
\EventLocation{Little Whinging, United Kingdom}
\EventLogo{}
\SeriesVolume{42}
\ArticleNo{23}
}

\arxiv{
  \author{
Ittai Abraham\\
VMWARE\\
\texttt{ittaia@cs.huji.ac.il}
\and
Danny Dolev%
\thanks{Danny Dolev is Incumbent of the Berthold Badler Chair in
  Computer Science.} 
\\
School of Computer Science and Engineering\\
The Hebrew University of Jerusalem\\
Jerusalem, Israel\\
\texttt{dolev@cs.huji.ac.il}
\and
Ivan Geffner%
\thanks{Supported in part by NSF grants IIS-1703846.}\\
Cornell University\\
\texttt{ieg8@cornell.edu}
\and
Joseph Y. Halpern%
\thanks{Supported in part by NSF grants IIS-1703846 and 
IIS-0911036, and ARO grant W911NF-17-1-0592, and a grant 
from Open Philanthropy.}
\\
Cornell University\\
\texttt{
halpern@cs.cornell.edu
}
  }
}

\begin{document}

\maketitle

\begin{abstract}
A mediator can help non-cooperative agents obtain an equilibrium
that may otherwise not be possible.  We study the
ability of players to obtain the same equilibrium without a
mediator,
using only \emph{cheap talk}, that is, nonbinding pre-play communication.
Previous work has considered this problem in a synchronous
setting.
Here we consider the effect of asynchrony on the problem, and
provide upper bounds for implementing mediators.
Considering asynchronous
environments
introduces new subtleties, including
exactly what solution concept is most appropriate and
determining what move is played if the
cheap talk goes on forever.
Different results
are obtained
depending on
whether the move after such ``infinite play'' is under the control of
the players or part of the description of the game.
\end{abstract}

\thispagestyle{empty}


\disc{
\smallskip
\noindent This paper is eligible for the best student paper award.
}

\section{Introduction}


Having a trusted mediator often makes solving a problem much easier.
For example, a problem such as Byzantine agreement becomes trivial
with a mediator: agents can just send their initial input to the
mediator, and the mediator sends the majority value back to all the
agents, which they then output.  Not surprisingly, the question of
whether a problem in a multiagent system that can be solved with a
trusted mediator can be solved by just the agents in the system,
without the mediator, has attracted a great deal of attention in
both computer science (particularly in the cryptography community)
and game theory.  In cryptography, the focus
has been
on \emph{secure multiparty computation} \cite{GMW87,yao:sc}.  Here it
is assumed that 
each agent $i$ has some private information $x_i$.  Fix functions
$f_1, \ldots, f_n$.  The goal is to have agent $i$ learn $f_i(x_1,
\ldots, x_n)$ without learning anything about $x_j$ for $j \ne i$
beyond what is revealed by the value of $f_i(x_1, \ldots, x_n)$.
With a trusted mediator, this is trivial: each agent $i$ just gives
the mediator its private value $x_i$; the mediator then sends each
agent $i$ the value $f_i(x_1, \ldots, x_n)$.  Work on multiparty
computation provides conditions under which this can be done
in a synchronous system~\cite{bgw,GMW87,SRA81,yao:sc} and in an
asynchronous system~\cite{BCG93,BKR94}.
In game theory, the focus has been on whether an equilibrium in a
game with a mediator can be implemented using what is called
\emph{cheap talk}---that is, just by players communicating among
themselves.

In the computer science literature, the interest has been in
performing
multiparty computation in the presence of possibly malicious
adversaries, who do everything they can to subvert the computation.
On the other hand, in the game theory literature, the assumption is
that players have preferences and seek to maximize their utility;
thus, they will subvert the computation iff it is in their best
interests to do so.
In \cite{ADGH06,ADH07} (denoted ADGH and ADH, respectively,
in the rest of the paper),
it was argued that it is
important to consider deviations by both rational players, who have
preferences and try to maximize them, and players that we can view as
malicious, although it is perhaps better to think of them as
rational players whose utilities are not known by the mechanism
designer (or other players). 
ADGH and ADH considered equilibria that are
\emph{$(k,t)$-robust}; roughly speaking, this means that the
equilibrium tolerates deviations by up to $k$ rational players,
whose utilities are presumed known, and up to $t$ players with
unknown utilities. Tight bounds were proved on the ability to
implement a $(k,t)$-robust equilibrium in the game with a
mediator using cheap talk in synchronous systems.
These bounds
depend on, among other things, (a) the relationship between $k$, $t$ and $n$, the total
number of players in the system; (b) whether players know the exact
utilities of the rational players;
\fullv{(c) whether there are broadcast channels
or just point-to-point channels; (d) whether cryptography is
available; and (e)}
\shortv{and (c)}
whether the game has a \emph{punishment
strategy}, where an $m$-punishment strategy is a strategy profile 
that, if used by all but at most $m$
players, guarantees that every player gets a worse outcome than they do
with the equilibrium strategy.
The following is a high-level overview of
results proved in the synchronous setting that will be of most
relevance here.
For these results, we assume that the communication with the mediator
is bounded, it lasts for at most $N$ rounds, and that the mediator
can be represented by an arithmetic circuit of depth $c$.
\begin{itemize}
\item[R1.] If $n > 3k + 3t$, then a mediator can be implemented using cheap
  talk; no punishment strategy is required, no knowledge of other
agents' utilities is required, and the cheap-talk protocol has
bounded running time $O(nNc)$, independent of the utilities.

\item[R2.] If $n > 2k+3t$, then a mediator can be implemented using cheap talk
    if there is a $(k+t)$-punishment strategy and the utilities of the
    rational players are known;
the cheap-talk protocol has expected running time $O(nNc)$.
%
(In R2, unlike R1, the cheap-talk game may be unbounded,
although it has finite expected running time.)



\fullv{
\item[R3.] If $n > k+3t$ then, assuming cryptography and polynomially-bounded
players, we can $\epsilon$-implement a mediator using cheap talk
(intuitively,
there is an implementation where players get utility within $\epsilon$
of what they could get by deviating).
}
\end{itemize}

\arxiv{
In ADH, lower bounds are presented that match the upper bounds above.
Thus, for 
example, it is shown that $n > 3k+3t$ is necessary in R1; if $n \le 3k+3t$,
then we cannot implement a mediator in general if we do not have a
punishment strategy or if the utilities are unknown.
The proofs of
the upper bounds}
\disc{The proofs of R1 and R2}
make heavy use of the fact that the setting is synchronous.
Here
we consider the impact of asynchrony on these results.
Once we introduce asynchrony, we must revisit the question of what
it even means to implement an equilibrium using cheap talk.  Notions
like
(Bayesian)
Nash equilibrium implicitly assume that all uncertainty
is described probabilistically.
Having a probability is necessary to talk
about an
agent's expected utility, given that a certain strategy profile is
played. If we were willing to put a distribution on how long
messages take to arrive and on when agents are scheduled to move,
then we could apply notions like Nash equilibrium without
difficulty. However, it is notoriously difficult to quantify this
uncertainty.
The typical approach used to analyze algorithms in the presence of
uncertainty that is not quantified probabilistically is to assume
that all the non-probabilistic uncertainty is resolved by the environment 
according to some strategy.  Thus, the environment uses some strategy to
decide
when each agent will be allowed to play and how long each message takes
to be delivered.
The algorithm is then proved correct no matter what
strategy the environment is following in some class of strategies.  For
example, we might restrict the environment's strategy to being \emph{fair},
so that every agent eventually gets a chance to move.  (See
\cite{HT} for a discussion of this approach and further
references.)

We follow this approach in the context of games.
Note that once we fix the environment's strategy, we have an ordinary
game, where uncertainty is quantified by probability.  In this setting,
we can consider what is called \emph{ex post equilibrium}.  A strategy
is an ex post equilibrium if it is an equilibrium
no matter what strategy the environment uses.
Ex post equilibrium is a strong notion, but,
as we show by example, it can often be attained with the help of a
mediator.  We believe that it is
the closest analogue to Nash equilibrium in an
asynchronous setting.%
\fullv{
\footnote{Because of the difficulty of attaining ex post equilibrium,
other solution concepts
have been considered in non-probabilistic settings,
such as \emph{minimax-regret equilibrium} \cite{BH04} and
\emph{maximin equilibria} \cite{AB06}, where each player uses a strategy that
is a best response in a minimax-regret (resp., maximin) sense to the
choices of the other players.}
}

Another issue that plays a major role in an asynchronous setting is
what happens if
the strategies of players result in some
players being \emph{livelocked}, talking
indefinitely without making a move
in the underlying game,
or in some players being \emph{deadlocked},
waiting indefinitely without moving in the underlying game.
%
We consider two approaches for dealing with this problem.
One way to decide what action to assign to a player that fails to
make a decision in the cheap-talk phase is the \emph{default-move
approach}. In this approach, as part of the description of the game,
there is a default move for
each player
which is imposed if that
player fails to explicitly make a move in the cheap-talk
phase.
Aumann and Hart \citeyear{AH03} considered a different approach,
which we henceforth call the AH approach, where a player's
strategy in the underlying game is a
function of the (possibly infinite) history of the player in the
cheap-talk phase.  We can think of this almost as a player writing a
will, describing what he would like to do (as a function of the
history) if the game ends before he has had a chance to move.

\arxiv{
We believe that both the AH approach and the default-move approach are
reasonable in different contexts.
The AH approach makes sense if the agent can leave 
instructions that will be carried out by an ``executor'' if the
cheap-talk game deadlocks.
But if we consider a game-theoretic
variant of Byzantine agreement, it seems more reasonable to say that if
a malicious agent can prevent an agent from making a move in finite
time, the agent should not get a chance to make a move after the
cheap-talk phase has ended.
}

Our results show that, in the worst case, the cost of asynchrony is
an extra $k+t$ in the bounds on $n$,
but we can sometimes save $k$ or even $k+t$ if there is a
punishment strategy or if we are willing to tolerate an $\epsilon$
``error''.  For example, with both the AH approach and the default-move
approach, if the utilities are not known, we can implement a mediator
using asynchronous cheap talk if $n> 4k+4t$.  Thus, compared to R1,
we need an extra $k+t$.
However, if we are willing to accept a small probability of error, so
that rather than implementing the mediator we get only an
$\epsilon$-implementation, and are also willing to accept
$\epsilon$-(k,t)-robustness (which, roughly speaking, means that
players get within $\epsilon$ of the best they could get), then we can
do this if $n > 3k+3t$, again, using both the AH approach and the
default-move approach. 

\commentout{
In going from R1 to R2, we must assume both that utilities are known and
that there is a punishment strategy.  In the synchronous setting, there
is no advantage in making just one of these assumptions; we continue
to require $n > 3k + 3t$.
In the asynchronous case, we get
different results when we uncouple the assumptions.  For example, with
the AH approach, if there is a punishment strategy, then we can
implement a mediator using asynchronous cheap talk if $n > 3k+4t$,
even without knowing the exact utilities.
Intuitively, what happens is that a good player needs to wait to get
some information; if the message does not arrive, then the waiting
player instructs his executor to carry out a punishment in his will.
Having a punishment does not seem to help in the default-move approach
(we still require that $n > 4k+4t$) unless the default move is a
punishment; if it is, then we can again implement a mediator if $n > 3k + 4t$.  If there is no punishment strategy but the utilities are
known, then with both the AH and default-move approach, we can
$\epsilon$-implement a mediator if $n > 3k+3t$
Finally, if we have both a punishment
strategy and utilities are known (as in R2), then we can
$\epsilon$-implement a mediator with the AH approach
(and the default-move approach if the default induces punishment)
if $n > 2k +4t$.}
Just as in the synchronous case, we can do better if we assume that
there is a punishment strategy and utilities are known (as in R2).
Specifically, with the AH approach, we can implement a mediator 
if $n > 3k +4t$ 
(compared to $n > 2k + 3t$ in the synchronous case, and can
$\epsilon$-implement a mediator if $n > 2k+3t$.  We use the punishment 
to deal with deadlock.   If a good player is waiting for a message
that never
arrives, then the waiting player instructs his executor to
carry out a punishment in his will.  Having a punishment does not seem
to help in the default-move approach 
unless the default move is a punishment; if it is, then we can
get the same results as with the AH approach.  

If there is a punishment strategy, these results significantly improve
those of Even, Goldreich, and Lempel \citeyear{EGL85}.  They provide a
protocol with similar properties, but the expected number of messages
sent is $O(1/\epsilon)$;
with a punishment strategy, a bounded number of messages is sent, with
the bound being independent of $\epsilon$.

\commentout{
To prove the upper bounds, we use the same general strategy as in ADGH,
which in turn goes back to the multiparty computation protocol of
Goldreich, Micali, and Wigderson
\citeyear{GMW87}.  We first simulate the computation of a circuit that
computes the values sent by the mediator, then share these values using
a secret-sharing protocol.  However, since we are dealing with rational
players, we must ensure that the secret-sharing protocol provides
appropriate incentives to the players to share their shares of the
secret.
}

\commentout{
For the circuit simulation, we could make use of the known protocols for
asynchronous multiparty computation \cite{BCG93,BKR94}.  But for these
protocols to provide perfect security,
we require that $n > 4t$ \cite{BKR94}.
While we do make use of these protocols for some of our
results (the one where the right-hand side has a $4t$ term), we can
sometimes do better if there is a punishment strategy.
we can
essentially force enough synchrony so as to be able to use synchronous
multiparty computation, for which we require only that $n> 3t$.
}


\fullv{
To do rational secret sharing in the asynchronous setting, we must
overcome a significant difficulty: it is possible that a player has
already received enough shares to compute the secret before it has
sent its share.
If $n$ is sufficiently large compared to $k$ and $t$, we can overcome
this difficulty using techniques already used in ADGH, which involve
using Reed-Solomon unique decoding \cite{J76} and check vectors, as in
Rabin and Ben-Or's \citeyear{RB89} \emph{information-checking protocol}.
However, with the AH approach, we can do better if there is a punishment
strategy, using a novel two-phase rational secret-sharing protocol.
Every round of this protocol has two phases. In the first phase a
value is shared amongst the players, but no one knows if this value
is the real secret or just a random number. In the second phase a
coin is shared
that indicates if the value seen in the first phase of that round is
the real secret. Players have no motivation to deviate in the
first phase since they will probably learn nothing and they may get
caught and punished. Players have no motivation to deviate in the
second phase either.  If a player deviates by sending an incorrect
value, then the player is likely to get caught and punished if the value
in the first phase is random.  On the other hand, if a player deviates
by sending nothing, the other players
will get enough values to learn the secret anyway; if a player deviates
by lying about the value, a problem will be detected, and the other players
will know
that the value sent in the first phase must have been the true value,

The synchronous model used in ADGH and ADH
involves two
assumptions: (1) all messages sent in round $k$ are received before the
beginning of round $k+1$, and (2) messages are sent and received
simultaneously.  We can decouple these assumptions to get a deeper
understanding of the power of asynchrony. Suppose that we consider a
model where every agent is scheduled to send
messages
in round $k$, and
the messages received
before the beginning of round $k+1$, but the
order in which the agents are scheduled and the messages are received is
determined by the environment.  Thus, for example, it is possible that an agent
will receive every other agents' round $k$ message before sending his
round $k$
messages.
We call this the \emph{partially synchronous}
setting;%
\footnote{This is the traditional synchronous model in a distributed system.}
it is a generalization of what AH call \emph{polite cheap
talk}, where in each round, only one agent talks.

Somewhat surprisingly,
as we show in the full paper,
using the techniques presented
here
(particularly the ``two-phase'' technique described above) we
can
get the same bounds in terms of $k$
and $t$ in the partial synchronous setting as in the fully synchronous
setting of ADGH.  That means that the bounds of ADGH still hold
in large-scale networks, where it is
reasonable to assume that message delays have some uncertainty but
there is a known upper bound on the delay.
}

\fullv{
The rest of this paper is organized as follows.  In
Section~\ref{sec:definitions}, we review all the relevant
definitions.
In Section~\ref{sec:upperbounds}, we
we give our upper bound results.
}

\section{Definitions}\label{sec:definitions}

{\bf Asynchronous games, mediator games, and cheap talk:}
We are interested in implementing mediators.  Formally, this means
we need to consider three games: an \emph{underlying game} $\Gamma$,
an extension
$\Gamma_{d}$ of $\Gamma$ with a mediator, and
an extension $\Gamma_{\ACT}$ of $\Gamma$ with (asynchronous)
cheap-talk.
We assume that $\Gamma$ is a \emph{normal-form Bayesian
game}: each player has a type $t$ taken from some type space $\T_i$,
such that there is a commonly known distribution on $\T \subseteq
\T_1 \times \cdots \times \T_n$, the set of types; each player $i$
chooses an action $a \in A_i$, the set of actions of agent $i$;
player $i$'s utility $u_i$ is determined
by the type profile of the players and the actions they take.
A strategy for player $i$ in the Bayesian game is
just a function $T_i$ to $A_i$, which tells player $i$ what to do,
given his type. If $A = A_1 \times \cdots \times A_n$, then a 
strategy profile $\vec{\sigma} = (\sigma_1, \ldots, \sigma_n)$ can be
viewed as a function $\vec{\sigma}: \T \rightarrow 
\Delta(A)$ (where, as usual, $\Delta(X)$ denotes the set of
probability distributions on $X$).

The basic notions of a game with a
mediator, a cheap-talk game, and implementation are standard in the
game-theory literature.   However, since we consider them in an
asynchronous setting, we must modify the definitions somewhat.

We first define \emph{asynchronous games}.
In an asynchronous game, we assume that players alternate making moves
with the environment---first the environment moves, then a player
moves, then the environment 
moves, and so on.  The environment's move consists of choosing a player $i$  
to move next and a set of messages in transit to $i$ 
that will be  delivered just before $i$ moves (so that $i$'s move can depend on
the messages $i$ receives).  The environment is subject to two constraints: 
all messages sent must eventually be delivered and, for all times $m$
and players 
$i$, if $i$ is still playing the game at time $m$, then there must be
some time $m' \ge m$ that $i$ is chosen to move.  We can
describe an asynchronous game by a game tree.  Associated with each
non-leaf node or 
history is either a player---the player whose move it is at that
node---or the environment (which can make a randomized move).  The nodes
where a player $i$ moves are further partitioned into
\emph{information sets}; intuitively, these are nodes that player
$i$ cannot tell apart.  We assume that the environment has complete information,
so that the environment's information sets just consist of the
singletons.
A \emph{strategy} for player $i$ is
a (possibly randomized) function from $i$'s information sets to
actions; we can similarly define a strategy for the environment.
We can essentially view the environment strategy as defining a
scheduler
(and thus we sometimes refer to an environment strategy as a scheduler).
\fullv{
The description 
of an asynchronous game includes a set of \emph{allowable strategies}
for the environment.  All allowable strategies have the properties
that all messages are eventually delivered and that no player is blocked
forever from moving, however, we could put additional restrictions on
allowable strategies.  For example, if the only allowable strategy is
the one where all players move at every step and all messages sent at
time $m$ are delivered by time $m+1$, then we have a synchronous game.
Similarly,
if time can be divided into rounds so that all players are scheduled in
each round, and all messages sent in round $m$ are delivered before the
beginning of round $m+1$, then we have a partially synchronous game.
}

For our results, we start with an $n$-player Bayesian game $\Gamma$ in normal
form (called the \emph{underlying game}), with $\{1,\ldots,n\}$
being the set of players, and then consider two types
of games that \emph{extend} $\Gamma$.  A game $\Gamma'$ extends
$\Gamma$ if the players have initial types from the same type space
as $\Gamma$, with the same distribution over types; moreover, in
each path of the game tree for $\Gamma'$, the players send and
receive messages, and perform at most one action from $\Gamma$. In a
history where each player makes a move from $\Gamma$, each player
gets the same utility as in $\Gamma$ (where the utility is a
function of the moves made and the types).  That leaves open the
question of what happens in a complete history of $\Gamma'$ where
some players do not make a move in $\Gamma$.  As we suggested in the
introduction, we consider two approaches to dealing with this.  In
the first approach, we assume that the description of $\Gamma'$
includes a function $M_i$ for each player $i$ that maps player $i$'s
type to a move in $\Gamma$.  In an infinite history $h$ where $i$ has
type $t$ and does not make a move in $\Gamma$, $i$ is viewed as
having made move $M_i(t)$.  We can then define each player's utility
in $h$ as above.  This is the \emph{default-move approach}.  In the
AH approach, we extend the notion of strategy so that
$i$'s strategy in $\Gamma'$ also describes what move $i$ makes in
the underlying game $\Gamma$ in
any
infinite history $h$ where $i$ has not made a move in $\Gamma$.
In the AH approach, $i$'s move in $h$ is under $i$'s
control; in the default-move approach, it is not.

\commentout{
We believe that both the AH approach and the default-move approach are
both reasonable in different contexts.  As we said in the introduction,
the AH approach makes sense if the agent can leave a ``will''; i.e.,
instructions as to what to do in the event of an infinite history, that
will be carried out somehow.  But if we consider a game-theoretic
variant of Byzantine agreement, it seems more reasonable to say that if
a ``malicious'' agent can prevent an agent from making a move in finite
time, the agent should not get a chance to make move after the
cheap-talk phase has ended.
}

Given an underlying Bayesian game $\Gamma$ (which we assume is
synchronous---the players move simultaneously), we will be interested
in two 
types of extensions.  A \emph{mediator game} extending $\Gamma$ is
an asynchronous game  where players can send messages to and receive
messages from a
mediator (who can be viewed as a trusted third party) as well as
making a move in $\Gamma$; ``good''
or ``honest''
players do not send messages to each
other, but ``bad'' players (i.e., one of the $k$ rational deviating players or
one of the $t$ ``malicious'' players with unknown utilities) may send
messages to each other as well as to the mediator.
We assume that the space of possible messages that can be sent in a
mediator game is fixed and finite.

In an asynchronous cheap-talk game extending $\Gamma$, there is no
mediator.  Players send messages to each other via asynchronous
channels, as 
well as making a move in $\Gamma$. We assume that each pair of agents
communicates 
over an \emph{asynchronous  secure private channel} (the fact
that
the channels
are secure means that an adversary cannot eavesdrop on
conversations between the players).
We assume
that each player can identify the sender of each message.
\arxiv{
Finally, we assume that in both the mediator game and the cheap-talk
game, when a player is first scheduled, it gets a signal that the game
has started (either an external signal from the environment, or 
a game-related message from another player or the mediator).}

\commentout{
{\bf Solution concepts:}
A strategy profile is an $\epsilon$-Nash equilibrium if no
player can gain more than $\epsilon$ by using a different strategy, given
that all the other players do not change their strategies. In
ADGH, we defined equilibria that are $\epsilon$-\emph{$(k,t)$-robust};
roughly speaking, this means that the equilibrium tolerates
deviations by up to $k$ rational players, whose utilities are
presumed known, and up to $t$ players with unknown utilities.
We also considered \emph{strong} $\epsilon$-$(k,t)$-robustness.
For a strategy to be robust says that it not the case that all the
players in the coalition gain by deviating; strong robusteness
requires that not even one of the players in the coalition can gain by
deviating.  
Here
we further extend this notion so that it applies to asynchronous
games,
by requiring that an equilibrium be $\epsilon$-$(k,t)$-robust (as defined in
ADGH), no matter what strategy is used by the environment.
See
Section~\ref{sec:adversary} for a formal definition and
Appendix~\ref{sec:def-more} for more details and intuition. 
}

{\bf Implementation:}
In the synchronous setting, a strategy profile $\vec{\sigma}'$ in a
cheap-talk $\Gamma_{\ACT}$ extending an underlying game $\Gamma$
\emph{implements} a strategy $\vec{\sigma}$ in  a mediator game
$\Gamma_d$ extending $\Gamma$ if $\vec{\sigma}$ and
$\vec{\sigma}'$ correspond to the same strategy in $\Gamma$; that is,
they induce the same function from $\T$ to $\Delta(A)$.
The notion of implementation is more complicated in an
asynchronous setting, because the probability on action profiles
also depends on the environment strategy.  Because $\Gamma_{\ACT}$
and $\Gamma_d$ are quite different games, the environment's strategies
in $\Gamma_{\ACT}$ are quite different from those in $\Gamma_d$.  So
we now say that $\vec{\sigma}'$ implements $\vec{\sigma}$ if the set
of distributions on actions profiles in $\Gamma$ induced by $\vec{\sigma}$
and all possible choices of environment strategy is the same as that
induced by $\vec{\sigma}'$ and all possible choices of environment strategy.
More precisely, let $\S_{\Gamma',e}$
and $\S_{\Gamma'',e}$ denote the
the set of environment strategies 
in $\Gamma'$ and $\Gamma''$, respectively.
A strategy $\sigma_e \in \S_{\Gamma',e}$ and 
a strategy profile $\vec{\sigma}$ for the players in
$\Gamma'$ together induce a function $(\vec{\sigma},\sigma_e)$ from
$\T$ to $\Delta(A)$.
A strategy profile $\vec{\sigma}'$ in $\Gamma'$
\emph{implements} a 
strategy profile $\vec{\sigma}''$ in $\Gamma''$ 
if $\{(\vec{\sigma}',\sigma_e'): \sigma_e' \in \S_{\Gamma',e} \} =
\{(\vec{\sigma}'',\sigma_e''): \sigma_e'' \in \S_{\Gamma'',e} \}$.
Since the outcome that arises if the players use a particular
strategy may depend on what the environment does, 
this says that the set of outcomes that can result if the players
use $\vec{\sigma}'$ is the same as the set of outcomes that can result
if the players use $\vec{\sigma}''$.

For some of our results, we cannot get an exact implementation; there
may be some error.
Given two discrete distributions $\pi$ and $\pi'$ on some space $S$,
the \emph{distance} between $\pi$ and $\pi'$, denoted
$\dist(\pi,\pi')$, is at most $\epsilon$ if $\sum_{s \in S}|\pi(s) -
\pi'(s)| \le 
\epsilon$.
\arxiv{As we observed earlier, in the mediator game and the cheap-talk game,}
Recall that 
a strategy profile $\vec{\sigma}$ for the players and a strategy
$\sigma_e$ for the environment together induce a mapping from type
profiles to $\Delta(A)$.
We lift the notion of distance to such function by defining
$\dist((\vec{\sigma},\sigma_e)),(\vec{\sigma}',\sigma_e')) = \max_{\vec{x}
  \in \T} \dist((\vec{\sigma},\sigma_e)(\vec{x}),(\vec{\sigma}',\sigma_e')(\vec{x}))$.
Say that $\vec{\sigma}'$ $\epsilon$-\emph{implements}
$\vec{\sigma}''$ if 
\begin{itemize}
  \item for all $\sigma_e' \in \S_{\Gamma',e}$ there exists
  $\sigma_e'' \in \S_{\Gamma'',e}$ such that
     $\dist((\vec{\sigma}', \sigma'_e),
        (\vec{\sigma''}, \sigma''_e)) \le \epsilon$; and 
\item for all $\sigma_e'' \in \S_{\Gamma'',e}$ there exists $\sigma_e'
  \in \S_{\Gamma',e}$ such that
  $\dist((\vec{\sigma}'', \sigma''_e),
  (\vec{\sigma}', \sigma'_e)) 
  \le \epsilon$. 
\end{itemize}
Note that $\vec{\sigma}'$ implements $\vec{\sigma}''$ iff 
$\vec{\sigma}'$ 0-implements $\vec{\sigma}''$.


The notion of implementation is quite strong.  For example, 
if 
$\vec{\sigma}'$ involves fewer rounds of communication than 
$\vec{\sigma}''$, there may be far fewer distinct schedulers in the
game involving $\vec{\sigma}'$ than in the game involving
$\vec{\sigma}''$.  Thus, we may not be able to recover the effect of
all possible schedulers.   (Indeed, for some of our results the
implementation needs to be quite long precisely in order to capture
all possible schedulers.)  
This suggests the following notion:
a strategy profile $\vec{\sigma}'$ in $\Gamma'$ \emph{weakly implements} a
strategy profile $\vec{\sigma}''$ in $\Gamma''$ 
if $\{(\vec{\sigma}',\sigma_e'): \sigma_e' \in \S_{\Gamma',e} \} \subseteq
\{(\vec{\sigma}'',\sigma_e''): \sigma_e'' \in \S_{\Gamma'',e} \}$.
Thus, if $\vec{\sigma}'$ weakly implements $\vec{\sigma}''$, then
every outcome of $\vec{\sigma}'$ is one that could also have arisen
with $\vec{\sigma}''$,  but the converse may not be true.
Specifically, there may be some behaviors of the environment with
$\vec{\sigma}''$ that cannot be simulated by $\vec{\sigma}'$.
As we shall see, this may actually be a feature: we can sometimes
simulate the effect of only ``good'' schedulers.  In any case,
note that in the synchronous setting, implementation and weak
implementation coincide.  
\arxiv{
We can also define a notion of weak $\epsilon$-implementation in the
obvious way; we leave the details to the reader.}

{\bf Termination:}  We will be interested in asynchronous games
where, almost surely, the honest players stop sending messages and
make a move in the underlying game.
In the mediator games that we consider, this happens after only a
bounded number of messages have been sent.
But even with this bound, there may not be a point in a
history when players know that they can stop sending messages;  
although a player $i$ may have moved in the underlying game, $i$ may still
need to keep checking for incoming message, and may need to respond to
them, to ensure that other players can make the appropriate move.

For some of our results, we must assume that, in the mediator game,
there comes a point when all honest players know that they have
terminated the protocol; they will not get further messages from the
mediator and can stop sending messages to the mediator, and should
make a move in the underlying game if they have not done so yet.
For simplicity, for these results, we restrict the honest players and
the mediator 
to using strategy profiles that have the following \emph{canonical form}:
Using a canonical strategy, player $i$ sends a message to the mediator
in response to a message from the mediator that does not include
``STOP'' if it has not halted, and these are the only messages that
$i$ sends, in addition to an initial message to the mediator.
If player $i$ gets a message from the mediator that includes ``STOP'',
then $i$
makes a move in the underlying game and halts.
We assume that, as long as the honest players and mediator follow
their part of the 
canonical strategy profile, there is a constant $r$ such that, no matter
what strategy the rational and malicious 
players and the environment use,  
the mediator sends each player $i$ at most $r$ messages in each
history, and the final message includes 
``STOP''.
We conjecture that the assumption that
players and mediator are using a strategy in
canonical form in the mediator game is without loss of generality;
that is, a $(k,t)$-robust strategy profile in a mediator game $\Gamma_d$
can be implemented by a $(k,t)$-robust strategy profile in
$\Gamma_d$ that is in 
canonical form.
\arxiv{However, we have not proved this conjecture yet.}
\commentout{
{\bf Cryptography}
Some of the results presented in the paper require the use of
cryptography.  In an asynchronous environment we do not have an
explicit way to bound the ability of players to use background
methods in order to try to break the cryptography. In order to
address that we can assume the Dolev-Yao~\citeyear{DY83} model,
where
the players can take any action other than breaking the
cryptography. Alternatively, the same way that we introduced default
actions that take place in a finite time, when a player doesn't take
an action in a finite time, we can assume that while our
asynchronous protocols use cryptography, there exits an upper bound
on the total expected execution time in a way that the security
parameters are chosen is such a way that the probability of breaking
the cryptography is negligible.
}

\section{Solution concepts}
In this section, we review the solution concepts introduced in ADGH
and extend them to asynchronous settings.

Note that in an
asynchronous game $\Gamma$, the utility of a player $i$ can depend not
only on the strategies of the agents, but on what the environment
does.  Since we consider an underlying game, a mediator game, and a
cheap-talk game, it is useful to include explicitly in the utility
function which 
game is being considered.  Thus, we write
$u_i(\Gamma, \vec{\sigma}, \sigma_d, \sigma_e, \vec{x})$ to denote the expected
utility of player $i$ in game $\Gamma$ when players
play strategy profile $\vec{\sigma}$,
the mediator plays $\sigma_d$,
the environment plays $\sigma_e$, and the type profile is $\vec{x}$.
We typically say ``input profile'' rather than ``type profile'', since
in our setting, the type of player $i$ is just $i$'s initial input.
Note that if $\Gamma$ is the underlying game, the $\sigma_e$ component
is unnecessary, since the underlying game is assumed to be
synchronous.
We occasionally omit the mediator strategy $\sigma_d$ when it is clear
from context. 

Given a type space $\T$, a set $K$ of players,
and $\vec{x} \in \T$, 
let $\T(\vec{x}_K) = \{\vec{x}\,': \vec{x}\,'_K = \vec{x}_K\}$. If
$\Gamma$ is a Bayesian game over type space $\T$, $\vec{\sigma}$ is a 
strategy profile in $\Gamma$, and $\Pr$ is the probability on the type space
$\T$, let $$u_i(\Gamma, \vec{\sigma},\sigma_e, \vec{x}_K) = \sum_{\vec{x}\,' \in
    \T(\vec{x}_K)} \Pr(\vec{x}\,' \mid \T(\vec{x}_K)) \;
u_i(\Gamma,\vec{\sigma},\sigma_e,\vec{x}').$$ Thus,
$u_i(\Gamma,\vec{\sigma},\sigma_e,\vec{x}_K)$ is $i$'s expected payoff
if everyone uses strategy $\vec{\sigma}$ and type profiles are
in $\T(\vec{x}_K)$.

{\bf $k$-resilient equilibrium:}
In a standard game, a strategy profile is a Nash equilibrium if no
player can gain any advantage by using a different strategy, given
that all the other players do not change their strategies.
The notion of $k$-resilient equilibrium extends Nash equilibrium to
allow for coalitions.

\begin{definition}\label{def:1} $\vec{\sigma}$ is a \emph{$k$-resilient
    equilibrium}
  (resp., \emph{strongly $k$-resilient equilibrium})
  in an asynchronous game $\Gamma$ if, for all subsets $K$
  of players with $1 \le |K| \le k$, all strategy profiles $\vec{\tau}_K$
  for the players in $K$, 
  all type profiles $\vec{x} \in \T$, and all
  strategies $\sigma_e$ of the environment,
  $u_i(\Gamma,(\vec{\sigma}_{-K},\vec{\tau}_K),\sigma_e, 
  \vec{x}_K) \le u_i(\Gamma,\vec{\sigma}, \sigma_e, \vec{x}_K)$ for
 some (resp., all)
  $i \in K$.%
\footnote{As usual, the strategy profile $(\vec{\sigma}_{-K}, \vec{\tau}_K)$
is the one where each player $i \in K$ plays $\tau_i$ and each
player $i \notin K$ plays $\sigma_i$.}
\end{definition}
Thus, $\vec{\sigma}$ is $k$-resilient if, no matter what the environment
does, no subset $K$ of at most $k$ players can all do better by
deviating, even if they share their type information (so that if the
true type is $\vec{x}$, the players in $K$ know $\vec{x}_K$).
It is strongly $k$-resilient if not even one of the players in $K$ can
do better if all the players in $K$ deviate.

\commentout{
\begin{definition}\label{def:2} $\vec{\sigma}$ is a \emph{strongly
        $k$-resilient
equilibrium} of an asynchronous game $\Gamma$ if for all subsets $K$
of players, and all type profiles $\vec{x} \in \T$, and all strategies $\sigma_e$
of the environment, it is not the case that there exists a strategy
profile $\vec{\tau}$ such that
$u_i(\Gamma,(\vec{\sigma}_{-K},\vec{\tau}_K),\sigma_e,\vec{x}_K) >
u_i(\Gamma,\vec{\sigma}, \sigma_e, \vec{x}_K)$ for some $i \in K$.
\end{definition}
}


\commentout{
Both of these definitions have a weakness: they implicitly assume
that the coalition members cannot communicate with each other beyond
agreeing on what strategy to use. As we show in ADH,
allowing communication can \emph{prevent} certain equilibria.  This
motivates the following definition.

\begin{definition}\label{def:3} $\vec{\sigma}$ is a \emph{(strongly) $k$-resilient
equilibrium} in an asynchronous game $\Gamma$ if $\vec{\sigma}$ is a
(strongly) $k$-resilient$'$ equilibrium in the cheap-talk extension
of $\Gamma$ (where we identify the strategy $\sigma_i$ in the game
$\Gamma$ with the strategy in the cheap-talk game where player $i$
never sends any messages beyond those sent according to $\sigma_i$).
\end{definition}

\noindent By considering the cheap-talk extension, we effectively
allow the players in the coalition arbitrary communication when they
deviate.
}

For some of our results we will be interested in equilibria that
are ``almost'' $k$-resilient, in the sense that no player in a
coalition can do more than $\epsilon$ better if the coalition its strategy, 
for some small $\epsilon$.

\begin{definition}\label{def:eps-equilibrium} For 
  $\epsilon > 0$, 
  $\vec{\sigma}$ is an \emph{$\epsilon$-$k$-resilient equilibrium}
(resp.,   \emph{strongly $\epsilon$-$k$-resilient equilibrium})
  if, for all subsets $K$ of players,
  all strategy profiles $\vec{\tau}_K$ for the players in $K$,
  all type profiles $\vec{x} \in \T$, and all
  strategies $\sigma_e$ of the environment,
we have
  $u_i(\Gamma,(\vec{\sigma}_{-K},\vec{\tau}_K),\sigma_e, \vec{x}_K) 
<
 u_i(\Gamma,\vec{\sigma},\sigma_e,\vec{x}_K) + \epsilon$ for
 some (resp., for all) $i \in K$.
\end{definition}
Note that we have ``$<  u_i(\Gamma,\vec{\sigma},\sigma_e,\vec{x}_K) +
\epsilon$'' here, not ``$\le$''; this means that a $0$-$k$-resilient
equilibrium is not a $k$-resilient equilibrium.  However, an
equilibrium is $k$-resilient iff it is $\epsilon$-$k$-resilient for
all $\epsilon > 0$.  We have used this slightly nonstandard definition
to make the statements of our theorems cleaner.

{\bf Robustness:}
A standard assumption in game theory is that utilities are
(commonly) known; when we are given a game we are also given each
player's utility. When players make decision, they can take other
players' utilities into account. However, in large systems, it seems
almost invariably the case that there will be some fraction of users
who do not respond to incentives the way we expect.
\arxiv{For example, in
a peer-to-peer network like Kazaa or Gnutella, it would seem that no
rational agent should share files. Whether or not you can get a file
depends only on whether other people share files; on the other hand,
it seems that there are disincentives for sharing (the possibility
of lawsuits, use of bandwidth, etc.). Nevertheless, people do share
files.  However, studies of the Gnutella network have shown almost
70 percent of users share no files and nearly 50 percent of
responses are from the top 1 percent of sharing hosts~\cite{AH00}.
}
\fullv{
  
One reason that people might not respond as we expect is that they
have utilities that are different from those we expect.  That is, if
we take $\vec{u}$ as characterizing the utilities we
\emph{expect} players to have,  what we are saying here is that
$\vec{u}$ may not correctly represent the utilities of all the
players. In the Kazaa example, it may be the case that some users
derive pleasure from knowing that they are the ones providing files
for everyone else.  In a computer network, inappropriate responses
may be due to faulty computers or faulty communication links. Or,
indeed, users may simply be irrational.  Whatever the reason, it}
\shortv{It}
seems important to design protocols that tolerate such unanticipated
behaviors, so that the payoffs of the users
\fullv{with ``standard'' utilities do not get affected by the
  nonstandard players using different strategies.
}
who follow the recommended strategy do not get affected by players who
deviate, provided that not too many deviate.

\begin{definition}
  A strategy profile $\vec{\sigma}$ is \emph{$t$-immune} in a game
  $\Gamma$ if, for all
  subsets $T$ of players with $|T| \leq t$, all strategy profiles $\vec{\tau}$,
all $i \notin T$, all type profiles $\vec{x} \in \T$, and all strategies $\sigma_e$
of the environment, we have
$u_i(\Gamma,(\vec{\sigma}_{-T},\vec{\tau}_T),\sigma_e, \vec{x}_T) 
\ge u_i(\Gamma,\vec{\sigma},\sigma_e,\vec{x}_T)$.
\end{definition}
Intuitively, $\vec{\sigma}$ is $t$-immune if there is nothing that
player in a set $T$ of size at most $t$ can do to give the players not
in $T$ a worse
payoff, even if the players in $T$ share their type information.

The notion of $t$-immunity and $k$-resilience address different
concerns.  For $t$-immunity, we consider the payoffs of the players
$K$. It is natural to combine both notions.  Given a strategy profile
$\vec{\tau}$, let $\Gamma^{T}_{\vec{\tau}}$ be the game which is
identical to $\Gamma$ except that the players in $T$ are fixed to
playing strategy $\vec{\tau}_T$.

\begin{definition}
  $\vec{\sigma}$ is a \emph{(strongly) $(k,t)$-robust} equilibrium
  in a game $\Gamma$
  if
  $\vec{\sigma}$ is $t$-immune and, for all subsets $T$ of players with
$|T| \leq t$ and all strategy profiles $\vec{\tau}$,
  $(\vec{\sigma}_{-T},\vec{\tau}_T)$ is a (strongly) $k$-resilient equilibrium of
    $\Gamma_{\vec{\tau}}^T$.
\end{definition}
\fullv{
Note that a (strongly) $k$-resilient equilibrium is a (strongly) $(k,0)$-robust
equilibrium.
The notion of $(0,t)$ resilience is somewhat in the spirit of
Eliaz's \citeyear{Eliaz00} notion of $t$ fault-tolerant
implementation, although Eliaz does not require $t$-immunity.}

We can define ``approximate'' notions of $t$-immunity and
$(k,t)$-robustness analogous to Definition~\ref{def:eps-equilibrium}:

\begin{definition}
    For $\epsilon > 0$, a  
strategy profile $\vec{\sigma}$ is \emph{$\epsilon$-$t$-immune} in
$\Gamma$ if, for all
subsets $T$ of players with $|T| \leq t$, all strategy profiles
$\vec{\tau}$,
all $i \notin T$, all type profiles $\vec{x} \in \T$, and all strategies $\sigma_e$
of the environment,
we have
$u_i(\Gamma,(\vec{\sigma}_{-T},\vec{\tau}_T),\sigma_e, \vec{x}_T)
> u_i(\Gamma,\vec{\sigma},\sigma_e,\vec{x}_T) - \epsilon$.
\end{definition}

\fullv{For $\epsilon$-$(k,t)$-robustness, we replace $t$-immunity and
 (strong) $k$-resilience for $\epsilon$-$t$-immunity and (strong)
  $\epsilon$-$k$-resilience, respectively:
  }

\begin{definition} For $\epsilon \ge 0$, $\vec{\sigma}$ is a
    \emph{(strongly) $\epsilon$-$(k,t)$-robust} equilibrium in
    $\Gamma$ if 
  $\vec{\sigma}$ is $\epsilon$-$t$-immune and, for all subsets $T$ of players with
$|T| \leq t$ and strategy profiles $\vec{\tau}_{T}$,
    $(\vec{\sigma}_{-T},\vec{\tau}_T)$ is a (strongly) $\epsilon$-$k$-resilient equilibrium of
    $\Gamma_{\vec{\tau}}^T$.
  \end{definition}



\section{Main theorems: formal statements}\label{sec:upperbounds}

In this section, we state our results formally.
Just as with \shortv{R1 and R2,} \fullv{R1--R3,}  for these results we assume that the
communication in the mediator game is bounded.  But since ``rounds''
is not meaningful in asynchronous systems, we express the bounds in
terms of number of messages.  Specifically, we assume that at 
most $N$ messages are sent in all histories of the mediator game, and
that the mediator can be represented by an arithmetic circuit with at
most $c$ gates.


We begin with a result that
is an analogue of R1 in the asynchronous setting.
%
We say that a game $\Gamma'$ is a utility variant of a game $\Gamma$
if $\Gamma'$ and $\Gamma$ have the same game tree, but the utilities
of the players may be different in $\Gamma$ and $\Gamma'$. We use the
notation $\Gamma(\vec{u})$ if we want to emphasize that $\vec{u}$ is
the utility function in game $\Gamma$. We then take $\Gamma(\vec{u}')$
to be the utility variant of $\Gamma$ with utility function
$\vec{u}'$.

One more technical comment before stating the theorems: in the
mediator game we can also view the mediator as a player (albeit one
without a utility function) that is following a strategy.
Thus, when we talk about a strategy profile that is a $(k,t)$-robust
equilibrium in 
the mediator game, we must 
give the mediator's strategy as well as the players' strategies.  We
sometimes write $\vec{\sigma}+\sigma_d$ if we want to distinguish the
players' strategy profile $\vec{\sigma}$ from the mediator's
strategy $\sigma_d$.  We occasionally abuse notation and drop the
$\sigma_d$ if it is clear from context, and just 
talk about $\vec{\sigma}$ being a $(k,t)$-robust equilibrium.


\begin{theorem}\label{thm:upperbound-default-no-punish}
If $\Gamma$ is a normal-form Bayesian game with $n$ players,
$\vec{\sigma}+\sigma_d$ is a strategy profile for the players and the
mediator in 
an asynchronous mediator game $\Gamma_d$ 
that extends $\Gamma$, and $n > 4k + 4t$,
then with both the default-move approach and the AH approach,
there exists
a strategy profile
$\vec{\sigma}_{\ACT}$ that implements
$\vec{\sigma} + \sigma_d$
in the asynchronous cheap-talk game $\Gamma_{\ACT}$ 
such that
for all utility variants $\Gamma_d(\vec{u}')$ of $\Gamma_d$,
if $\vec{\sigma} + \sigma_d$ is a (strongly) $(k,t)$-robust equilibrium
in $\Gamma_d(\vec{u}')$, then $\vec{\sigma}_{\ACT}$ is a (strongly)
$(k,t)$-robust 
equilibrium in $\Gamma_{\ACT}(\vec{u}')$, and
the number of messages sent in a history of $\vec{\sigma}_{\ACT}$ is
$O(nNc)$, independent of $\vec{u}'$.  
\end{theorem}

The proof of Theorem~\ref{thm:upperbound-default-no-punish} uses
ideas from the multiparty computation protocol of Ben-Or, Canetti, and
Goldreich \citeyear{BCG93} (BCG from now on).  Our construction
actually needs stronger properties than these provided by BCG; we
show that we can get protocols with these stronger properties in a
companion paper \cite{GH18}; see Section~\ref{sec:simulation} for
further discussion.


We can obtain better bounds if we are willing to accept
$\epsilon$-equilibrium,
using ideas due to Ben-Or, Kelmer, and Rabin \citeyear{BKR94}.

\begin{theorem}\label{thm:upperbound-default-no-punish1}
If $\Gamma$ is a normal-form Bayesian game with $n$ players,
$\vec{\sigma} + \sigma_d$ is a strategy profile for the players and mediator
in an asynchronous mediator game $\Gamma_d$ 
that extends $\Gamma$,
$M > 0$,
and $n > 3k + 3t$,
then with both the default-move approach and the AH approach, for all
$\epsilon > 0$, there exists
a strategy profile $\vec{\sigma}_{\ACT}$ in
the asynchronous cheap-talk game
$\Gamma_{\ACT}$ 
that
$\epsilon$-implements $\vec{\sigma}$
such that
for all utility variants $\Gamma_d(\vec{u}')$ of $\Gamma_d$
bounded by $M/2$ (i.e., where the range of $u_i'$ is contained in
$[-M/2,M/2]$),  
if $\vec{\sigma} + \sigma_d$ is a (strongly) $(k,t)$-robust
equilibrium 
in $\Gamma_d(\vec{u}')$, then $\vec{\sigma}_{\ACT}$ is a
(strongly) 
$\epsilon$-$(k,t)$-robust 
equilibrium in $\Gamma_{\ACT}(\vec{u}')$, and
the number of messages sent in a history of
$\vec{\sigma}_{\ACT}$ is $O(nNc)$, independent of $\vec{u}'$. 
\end{theorem}


If we have a punishment strategy and utilities are known, we can do
better with the AH approach.
\commentout{
The idea is to use the punishment to force the players to act in a
more synchronized manner.
We proceed in virtual rounds;
in each round, each player waits for an ``appropriate'' response from
$n-t-1$ other players.  (Inappropriate messages are assumed to come from
malicious players, so are ignored.)  If enough players do not
send an appropriate message, then there will be a deadlock.
Players write their will so that deadlocks result in punishment.
}
To make this precise, we need the definition of an $m$-punishment strategy
\cite{ADGH06} (which generalizes the notion of punishment strategy defined by 
Ben Porath \citeyear{Bp03}).
Before defining this carefully, note that in an asynchronous setting
(i.e., in the mediator game and the cheap-talk game, but not in the
underlying game),
the utility of players depends on the environment's strategy as well
as the players' strategy profile and the players' type profile.

\begin{definition}
  \commentout{
  Given a type space $\T$, a set $K$ of players, and $\vec{x} \in \T$, let
 $\T(\vec{x}_K) = \{\vec{x}': \vec{x}'_K = \vec{x}_K\}$.
If $\Gamma$ is an underlying game over type space $\T$, $\vec{\sigma}$ is a
strategy profile in $\Gamma$, and $\Pr$ is the probability on the
type space $\T$, let $$u_i(\Gamma,\vec{\sigma},\vec{x}_K)
= \sum_{\vec{x}\,' \in \T(\vec{x}_K)} \Pr(\vec{x}\,' \mid \T(\vec{x}_K))
u_i(\Gamma,\vec{\sigma},\vec{x}).$$
Thus, $u_i(\Gamma,\vec{\sigma},\sigma_e,\vec{x}_K)$ is $i$'s expected payoff
if the players use strategy profile $\vec{\sigma}$
and type profiles are in $\T(\vec{x}_K)$.
We can similarly define $u_i(\Gamma_d,\vec{\sigma},\sigma_e,\vec{x}_K)$ for the
mediator game, where we do consider an environment.
}
If $\Gamma'$ is an extension of an underlying game $\Gamma$,
a strategy profile $\vec{\rho}$ in $\Gamma$ is a \emph{$k$-punishment
strategy with respect to a strategy profile $\vec{\sigma}'$ in $\Gamma'$} if
for all subsets $K$ of players with $1 \le |K| \leq k$, all
strategy profiles $\vec{\sigma}$ in $\Gamma$,
all strategies $\sigma_e$ for the environment,
all type profiles $\vec{x} \in \T$,  
and all players $i \in K$, we have
$$u_i(\Gamma',\vec{\sigma}',\sigma_e,\vec{x}_K) >
u_i(\Gamma,(\vec{\sigma}_{K},\vec{\rho}_{-K}), \vec{x}_K).$$
\end{definition}
%
Thus, if $\vec{\rho}$ is a $k$-punishment strategy with respect to
$\vec{\sigma}'$, if all but $k$ players play their part of
$\vec{\rho}$ in the underlying tame, then all of the remaining players
will be worse off than 
they would be
in $\Gamma'$ if everyone had played $\vec{\sigma}'$,
no matter what they do in the underlying game.

\commentout{
Intuitively, a punishment strategy can be used to threaten rational
players to follow the protocol in $\Gamma'$. However, if the good
players' action in $\Gamma'$ is determined by the punishment strategy,
this guarantees only that the deviating players in $K$ will get a
payoff worse than their expected equilibrium outcome.  Thus,
for some local histories, deviating players might actually be better
off if the remaining players play the punishment strategy (where
player $i$'s \emph{local history} at time $t$ consists of $i$'s
initial state and the sequence of messages sent and received by $i$ up
to time $t$). That is,
for some local histories, the threat of playing the punishment
strategy is not effective.  This can be an issue in some games, as we
show in Section~\ref{sec:punish}.  This observation motivates the
following stronger notion.

\begin{definition}
  Let $u_i(\Gamma, \vec{\sigma}, \sigma_e, \vec{x} \mid \vec{h})$ denote the
  expected payoff of player $i$ in game $\Gamma$ if players play
  strategy profile $\vec{\sigma}$, the environment plays $\sigma_e$,
  and the type profile is $\vec{x}$, conditional on the players having 
local history profile $\vec{h}$. If $\Gamma'$ is an extension of an
underlying game $\Gamma$, 
a strategy profile $\vec{\rho}$ in $\Gamma$ is a \emph{strong $k$-punishment
strategy with respect to a strategy profile $\vec{\sigma}'$ in $\Gamma'$} if
for all subsets $K$ of players with $1 \le |K| \leq k$, all
strategy profiles $\vec{\sigma}$ in $\Gamma$,
all strategies $\sigma_e$ for the environment,
all type profiles $\vec{x} \in \T$,  
and all players $i \in K$, and all local history profiles for players
in $K$ we have 
$$u_i(\Gamma',\vec{\sigma}',\sigma_e,\vec{x}_K \mid \vec{h}_K) >
u_i(\Gamma,(\vec{\sigma}_{K},\vec{\rho}_{-K}), \vec{x}_K).$$
\end{definition} 
}

\begin{theorem}\label{thm:punish}
If $\Gamma$ is a normal-form Bayesian game with $n$ players,
$\vec{\sigma}+\sigma_d$ is a strategy profile
in canonical form
for the players and mediator in an asynchronous mediator game $\Gamma_d$ 
that extends $\Gamma$,  there is a $(k+t)$-punishment strategy with
respect to $\vec{\sigma}+\sigma_d$, and $n
> 3k+4t$, then 
with the AH approach,
there exists
a strategy profile
$\vec{\sigma}_{\ACT}$  that 
implements
$\vec{\sigma} + \sigma_d$ in the asynchronous cheap-talk game $\Gamma_{\ACT}$,
and if 
$\vec{\sigma}+\sigma_d$ is a (strongly) $(k,t)$-robust
equilibrium in $\Gamma_d$,
then $\vec{\sigma}_{\ACT}$ is a (strongly)
$(k,t)$-robust
equilibrium  in $\Gamma_{\ACT}$.
If 
there exists a strong $(k+t)$-punishment strategy or
we require only that $\vec{\sigma}_{\ACT}$ is a weak
implementation, then
the number of messages in a history of
$\vec{\sigma}_{\ACT}$ is $O(nc)$ (and $\vec{\sigma}_{\ACT}$ continues
to be a (strongly) $(k,t)$-robust equilibrium in $\Gamma_{\ACT}$ if
$\vec{\sigma}$ is a (strongly) $(k,t)$-robust equilibrium in $\Gamma_d$).
\end{theorem}

Note that in Theorem~\ref{thm:punish}, the running time of the
algorithm is significantly affected by whether we want
$\vec{\sigma}_{\ACT}$ to implement $\vec{\sigma}$ or whether a weak
implementation suffices.

If we assume both that there is a $(2k+2t)$-punishment strategy and
that utilities are 
known, we can get an analogue to R2, but with an $\epsilon$ error.

\begin{theorem}\label{thm:punish-eps}
If $\Gamma$ is a normal-form Bayesian game with $n$ players,
$\vec{\sigma}+\sigma_d$ is a strategy profile
in canonical form
for the players and mediator in an asynchronous mediator game $\Gamma_d$ 
that extends $\Gamma$, there is a
$(2k+2t)$-punishment strategy with respect to $\vec{\sigma} + \sigma_d$, and $n >
2k + 3t$, then with the AH approach, for all $\epsilon > 0$ there
is a strategy
$\vec{\sigma}_{\ACT}$ that 
$\epsilon$-implements 
$\vec{\sigma}$
in the asynchronous cheap-talk game $\Gamma_{\ACT}$
such that if
$\vec{\sigma} + \sigma_d$ is a (strongly) $(k,t)$-robust equilibrium in $\Gamma_d$, then 
$\vec{\sigma}_{\ACT}$ is a (strongly)  $\epsilon$-$(k,t)$-robust
equilibrium in $\Gamma_{\ACT}$,
and the number of messages sent in a history of
$\vec{\sigma}_{\ACT}$ is $O(n2^Nc)$.
If 
there exists a strong $(k,t)$-punishment strategy or
we require only that $\vec{\sigma}_{\ACT}$ is a weak
implementation, then
the number of messages in a history of
$\vec{\sigma}_{\ACT}$ is $O(nc)$ (and $\vec{\sigma}_{\ACT}$ continues
to be a (strongly) $(k,t)$-robust equilibrium in $\Gamma_{\ACT}$ if
$\vec{\sigma}$ is a (strongly) $(k,t)$-robust equilibrium in $\Gamma_d$).
\end{theorem}

\shortv{We prove these results using ideas in the spirit of ADGH, but
  much more care must be taken to deal with asynchrony.
  Among other
  things, we need stronger security guarantees than are traditionally
  provided for multiparty 
  communication; see Section~\ref{sec:simulation} for details.}
\disc{We provide a detailed proof of
  Theorem~\ref{thm:upperbound-default-no-punish} in
Section~\ref{sec:thm4.1proof}, then sketch the remaining arguments.
Details can be found in the full paper \cite{ADGH18}.}
\arxiv{We provide proofs of all the results in Section~\ref{sec:proofs}.}

\fullv{
Finally, we can strengthen Theorem~\ref{thm:punish} by taking into
account the utility functions of the rational agents (note that
Theorem~\ref{thm:AH1} does not hold for all utility variants): 

\begin{theorem}\label{thm:AH1}
If $\Gamma$ is a normal-form  Bayesian game with $n$ players,
$\vec{\sigma}$ is a strategy for the players and mediator in
a game $\Gamma_{d}$ with 
a mediator $d$ that extends $\Gamma$, there is a
$(2k+2t)$-punishment strategy with respect to $\vec{\sigma}$, and $n >
2k+4t$, then with the AH approach, there exists an asynchronous
cheap-talk game $\Gamma_{\ACT}$ and a strategy
$\vec{\sigma}_{\ACT}$ that 
implements
$\vec{\sigma}+ \sigma_d$ such that if
$\vec{\sigma}$ is a (strongly) $(k,t)$-robust equilibrium
in $\Gamma_d$, then $\vec{\sigma}_{\ACT}$ is a (strongly) $(k,t)$-robust
equilibrium in $\Gamma_{\ACT}$.
\end{theorem}
}


\section{$t$-bisimulation and $t$-emulation}\label{sec:simulation}

To construct the cheap-talk protocol for
Theorems~\ref{thm:upperbound-default-no-punish} and
\ref{thm:upperbound-default-no-punish1}, we use ideas from a companion
paper \cite{GH18}, where we provide constructions that extend the
security guarantees given by the multi-party computation protocols for 
the synchronous case by Ben-Or, Goldwasser, and Wigderson
\citeyear{bgw} (BGW from now on) and the asynchronous case by
Ben-Or, Canetti, and Goldreich
\citeyear{BCG93} (BCG from now on).
We briefly review the main details here.

BGW/BCG
show that if a function $f$ of $n$ inputs provided by $n$
players can be computed using a mediator, then it can be computed by
the players without the mediator and without revealing any information
beyond the function value, even when some of the players are 
malicious. BGW deals 
with the synchronous case and tolerates up to $n/3$ malicious players,
while BCG deals with the asynchronous case and tolerates up to
 $n/4$. 
The notion of not revealing any information is made precise by
defining a set of \emph{ideal}
distributions over possible values of the function,
and ensuring that the \emph{real}
distribution  
is
identically distributed to one of those (see BGW and BCG
for formal definitions and details).

We can view a mediator game as computing an action
profile in the underlying game; the ideal distributions are the possible
distributions over action profiles when the honest players
play 
their component of the $(k,t)$-robust equilibrium strategy profile
in the mediator game.  BCG's protocol then essentially gives us a
strategy in the cheap-talk game.  However, the BCG protocol
is not sufficient for our 
purposes for two reasons: it does not guarantee that the real protocol is
an implementation of the ideal protocol in the sense of the
definition in Section~\ref{sec:definitions}
(although it does suffice for \emph{weak} implementation),
nor does it  guarantee 
that the protocol is a $(k,t)$-robust equilibrium.
To prove these stronger results, we show that
$\vec{\sigma}_{\ACT}$ can be constructed so as to
satisfy some additional security properties, which we now define.


\begin{definition} [$t$-bisimulation]
  Take an \emph{adversary} $A$ to be a pair $(\vec{\tau}_T,\sigma_e)$
    consisting of a strategy for the malicious players and an
    environment strategy.
  Let $O(\vec{\pi} + \pi_d, A, \vec{x})$ be the
  distribution over outputs when running strategy $\vec{\pi}$ with
 adversary $A = (\tau_T, \sigma_e)$.  Protocol $\vec{\pi}'$
 \emph{$t$-bisimulates}  
$\vec{\pi} + \pi_d$ if,
  for all $T$ with $|T| \le t$ and inputs $\vec{x}$: 
\begin{itemize}
  \item for all adversaries $A = (\vec{\tau}_T,\sigma_e)$, there
exists an adversary $A' = (\vec{\tau}_T',\sigma_e')$  such that
  $O(\vec{\pi} + \pi_d, A, 
  \vec{x})$ and $O(\vec{\pi}', A', \vec{x})$ are
    identically distributed; 
  \item for all adversaries $A' = (\vec{\tau}'_T,\sigma_e')$, there
        exists an adversary $A = (\vec{\tau}_T,\sigma_e)$ such that 
  $O(\vec{\pi} + \pi_d, A, 
  \vec{x})$ and $O(\vec{\pi}', A', \vec{x})$ are
  identically distributed. 
\end{itemize}
\end{definition}



For one direction of the simulation, we need an even tighter
correspondence between deviations in the cheap-talk game and
deviations in the mediator game. This is made
precise in the following definition.

\commentout{
There are two significant differences between $t$-bisimulation and
$t$-emulation. 
As the name suggest, with $t$-bisimulation, we require simulation in both
directions (for every $\vec{\tau}$ and $\sigma_e$ there is a
$\vec{\tau}'$ and $\sigma_e'$, and vice versa); for emulation, we have
only one direction.
On the other hand, with $t$-emulation, 
the strategy 
$\vec{\tau}_i'$ depends only on $\vec{\tau}_i$ and
$\sigma_e$, whereas with
$t$-bisimulation, 
$\vec{\tau}_i'$ can also depend on $\sigma_e'$ and
all the strategies in $\vec{\tau}$.
}
\begin{definition}[$t$-emulation]
The protocol $\vec{\pi}'$
\emph{$t$-emulates}
$\vec{\pi}$ if,
for every scheduler $\sigma'_e$, there exists a function $H$ from
strategies to strategies
such that $H(\pi'_i) = \pi_i$ for all players $i$ and, 
for all sets $T$ of players with $|T| \le t$ 
and all adversaries $A' = (\vec{\tau}'_T, \sigma'_e)$, there exists an
adversary $A = (\vec{\tau}_T, \sigma_e)$ such that, for all
input profiles $\vec{x}$, $O(\vec{\pi},A, \vec{x})$ and
$O(\pi',A', \vec{x})$ are identically distributed, where 
$\vec{\tau}_T = H(\vec{\tau}'_{-T})$ (and we take
$H(\tau'_1,\ldots,\tau'_m) = (H(\tau_1), \ldots, H(\tau'_m))$).
\end{definition}

There are two significant differences between $t$-bisimulation and
$t$-emulation. 
As the name suggests, with $t$-bisimulation, we require simulation in both
directions (for every $\vec{\tau}$ and $\sigma_e$ there is a
$\vec{\tau}'$ and $\sigma_e'$, and vice versa); for emulation, we have
only one direction.
On the other hand, with $t$-emulation, the strategy 
$\vec{\tau}_i'$ depends only on $\vec{\tau}_i$ and
$\sigma_e$, whereas with
$t$-bisimulation, 
$\vec{\tau}_i'$ can also depend on $\sigma_e'$ and
all the strategies in $\vec{\tau}$.

Note that 0-bisimulation is equivalent to implementation, while
0-emulation 
is equivalent to 
weak implementation.  Implementation and
weak implementation consider only what happens when there is no
malicious behavior; bisimulation and emulation generalize these
notions by taking malicious behavior into account.  For some of our
results (specifically, Theorems~\ref{thm:upperbound-default-no-punish} and
\ref{thm:upperbound-default-no-punish1}), we 
use these notions, and show that they are achievable under the
conditions of these theorems.
In fact, although we don't need it for our proof, we can show that we
can get $t$-bisimulation and $t$-emulation under the conditions of
Theorems~\ref{thm:punish} and~\ref{thm:punish-eps} as well.  We briefly
comment on how this can be done when we present the proof.

We have required schedulers to deliver each message eventually.
Because we assume that protocols in the mediator game are bounded,
all protocols in the mediator game must terminate.  This means that
we can't hope to emulate a protocol in the cheap-talk game that
deadlocks.  (We assume that if the protocol deadlocks, it has a
special output that we denote $\bot$.  Given our constraints, we can
never get an output of $\bot$ in the mediator game.)  To deal with this
situation, we relax this
requirement on schedulers somewhat, but only in the mediator game.
We take a \emph{relaxed scheduler} to be one that may never deliver
some messages.  However, we require that if the mediator sends several
messages at the same step, then a relaxed scheduler either delivers
all of them or 
none of them.  (There is no requirement on messages sent by the
players, since they send messages only to the mediator, and we can
assume without loss of generality that they send only one message at
each step.)  
We can define \emph{relaxed $t$-emulation} just as we defined
$t$-emulation, except that we now allow the scheduler $\sigma_e$ in
the definition to be a relaxed scheduler.
Finally, we define \emph{$(t,t')$-emulation} just as we defined relaxed
$t$-emulation except that $\sigma_e$ must be non-relaxed 
if $|T| \le t'$. 


\commentout{
\begin{definition}[Enhanced $t$-simulation]
Given two protocols $\pi$ and $\pi'$ for $n$ players, we say that
$\pi'$ enhanced $t$-simulates $\pi$ if for every scheduler $\sigma'_e$
for $\pi'$, there exists a function $H$ from strategies in $\pi'$ to
strategies in $\pi$ with $H(\pi'_i) = \pi_i$ for all $i$, such that
for all sets $T = \{a_1, \ldots, a_m\}$ of malicious players with $m
\le t$ and all strategy profiles $\vec{\tau}'_T = \{\tau'_{a_1},
\ldots, \tau'_{a_m}\}$ for players in $T$, there exists an enhanced
scheduler $\sigma_e$ for $\pi$ such that, for all input profiles $\vec{x}$,
$O(\pi, T, \vec{\tau}_T, \sigma_e, \vec{x})$ and $O(\pi', T,
\vec{\tau}'_T, \sigma'_e, \vec{x})$ are identically distributed, where
$\vec{\tau}_T = \{H(\tau'_{a_1}), H(\tau'_{a_2}), \ldots,
H(\tau'_{a_m})\}$. 
\end{definition}
}

We need a further property to deal with protocols that involve
punishment strategies.  For a punishment strategy to be effective, all
the honest players have to play it.  In our protocols, the punishment
strategy is played when there is a deadlock (so some players never
terminate); that is, the punishment strategy is in the honest players'
``wills''.  Thus, we want it to be the case that either none of the
honest players terminate (in which case the punishment strategy will
be effective) or all of them terminate; we do not want it to be the
case that only some of the honest players terminate.

\begin{definition}[$t$-cotermination]
A protocol $\vec{\pi}$ \emph{$t$-coterminates} if, for all schedulers
$\sigma_e$, all subsets $T$ of at most $t$ players, 
and all strategy profiles $\vec{\tau}_T$ for the players in $T$, in
all histories of the protocol $(\vec{\pi}_{-T},\vec{\tau}_T,
\sigma_e)$, either all the players 
not in $T$
terminate or none of them do.
\end{definition}

For some of our results, we need ``approximate'' versions of $t$-bisimulation,
$t$-emulation, relaxed $t$-emulation, and $t$-cotermination that allow 
an $\epsilon$ probability of error. 
For $t$-bisimulation, $t$-emulation, 
$(t,t')$-emulation, and relaxed $t$-emulation, this
means that the distance between the distribution over outputs in the
cheap-talk game and the distribution in the mediator game is at most
$\epsilon$ (where the notion of distance is that used in the
definition of $\epsilon$-implementation in Section~\ref{sec:definitions})
while for $t$-cotermination it means that the
property holds with probability $1 - \epsilon$. 
We call these properties
$(\epsilon, t)$-bisimulation, $(\epsilon, t)$-emulation,
relaxed $(\epsilon,t)$-emulation and $(\epsilon,
t)$-cotermination.
%
\arxiv{

}  
In \cite{GH18}, the following results are proved:
\begin{theorem}\label{thm:errorless-simulation}
Given a mediator game $\Gamma_d$ extending $\Gamma$ and a strategy
profile $\vec{\sigma} + 
\sigma_d$, there exists a strategy profile $\vec{\sigma}_{\ACT}$ for
$\Gamma_{\ACT}$ such that $\vec{\sigma}_{\ACT}$ 
$t$-coterminates, $t$-emulates (resp., relaxed $(t,t')$-emulates), and
$t$-bisimulates $\vec{\sigma} + \sigma_d$ if $t < n/3$, $t < n/4$
(resp., $t < n/3$ and $t' < n/4$), and $t < n/4$ respectively, 
and the expected number of messages in histories of
$\vec{\sigma}_{\ACT}$ is $O(nNc)$. 
\end{theorem}
\begin{theorem}\label{thm:error-simulation}
    Given a mediator game $\Gamma_d$ extending $\Gamma$, a strategy
    profile $\vec{\sigma} + 
\sigma_d$ in $\Gamma_d$, and a real number $\epsilon \in (0,1]$, there
  exists a strategy profile $\vec{\sigma}_{\ACT}$
in $\Gamma_{\ACT}$ such that $\vec{\sigma}_{\ACT}$
\commentout{
$(\epsilon,t)$-emulates (resp.,
relaxed $(\epsilon,t)$-emulates, $(\epsilon,t)$-coterminates)
and $(\epsilon,t)$-bisimulates $\vec{\sigma} + \sigma_d$ 
if $t < n/3$ (resp., $t < n/2$, $t < n/2$),
}
$(\epsilon,t)$-coterminates, $(\epsilon, t)$-emulates (resp.,
$(\epsilon, t, t')$-emulates), and $(\epsilon, t)$-bisimulates
$\vec{\sigma} + \sigma_d$ if $t < n/2$, $t < n/3$ (resp., $t < n/2$
and $t' < n/3$) and $t < n/3$ respectively, 
and the expected number of messages in histories of
$\vec{\sigma}_{\ACT}$ is $O(nNc)$. 
\end{theorem}

\section{Proofs of the main theorems}\label{sec:proofs}

\subsection{Coordination between the environment and malicious
  players}\label{sec:adversary} 

Before proving the main results, it is useful to understand some of
the implication of $(k,t)$-robustness, particularly when it comes to
the interactions between the environment and the malicious and rational
players. 
The definition of $(k,t)$-robustness requires that the rational
players have no profitable deviation no matter what the malicious
players and the environment do.  It may seem \emph{a priori} that the
malicious players, the rational players, and the environment all act
independently, but in fact, we can assume without loss of generality
that they are all under the control of a single adversary.
Clearly rational players can coordinate by sending messages to
each other.
  But it may not be obvious that the
malicious and rational players can also coordinate with the environment
so that, for example, 
the malicious and rational players can
act knowing who will be scheduled when and the environment can
schedule rational and malicious players based on their inputs.
Nevertheless, this follows from the fact that $(k,t)$-robustness must
hold for all schedulers.
We prove this by showing that rational and malicious players can
effectively communicate with the environment, even though the
environment cannot read messages.


To see that a player $i$ can communicate with the environment, recall
that we have assumed that the message space is finite,
say $\{m_0, \ldots, m_M\}$.  
Immediately after sending $m_j$, $i$ sends $j$ empty messages to
itself.  So, even though the environment cannot read the messages, it
will know that $i$ received $m_j$.%
\arxiv{
\footnote{We can encode $m_j$ using using fewer messages by having $i$
  send message to other agents as well as itself. Our goal here is
  not to minimize the number of messages, but to show that
  communication with the environment is possible.}}
(Clearly the environment will also know who sent the message,
since the environment delivered the message.)

The rational and malicious players know the environment's protocol
(and thus know when a message that is sent will be delivered).  Thus,
it suffices for the environment to tell the non-honest players when
the $k$th message is sent, who sent it, and who the intended recipient
is.  All the non-honest players $i$ initially send themselves 
\commentout{
$(n+1)(n+2)$ empty messages.  
If the first message was sent by player $j_1$ to player $j_2$
(treating the mediator as player $n+1$ in the mediator game), 
the the environment delivers $(n+1)j_1 +j_2$ empty messages.
Then player $i$ sends another $(n+1)(n+2)$ empty messages, allowing
the environment to encode the sender and receiver of the next message,
if there is one, and so on.
}
$(n+1)^2$ 
empty messages.
If the first message was sent by player $j_1$ to player $j_2$ 
(treating the mediator as player 0 in the mediator game), then the
environment delivers $(n+1)j_1 + j_2$ of these
empty messages.
Then player $i$ sends another 
$(n+1)^2$ 
 empty messages, allowing the
environment to encode the sender and receiver of the next message, if
there is one. 

\commentout{
This shows that rational players, malicious players and the environment
can communicate freely if they want to, which means that we can assume
without loss of generality that they are a single entity.
We use the notation $A = (\vec{\tau}_T, \sigma_e)$ to
denote an adversary that controls a subset $T$ of players and the
scheduler, in which players in $T$ deviate by following strategy
profile $\vec{\tau}_T$ and the scheduler follows strategy
$\sigma_e$.}
\commentout{
We write $u_i(\Gamma, \vec{\sigma}, \sigma_e, \vec{x})$ to denote the expected
utility of player $i$ in game $\Gamma$ when players
play strategy profile $\vec{\sigma}$, 
the environment plays $\sigma_e$, and the input profile is $\vec{x}$.    

A strategy profile is \emph{$t$-immune} if, for all
$t$-adversaries $A = (\vec{\tau}_T,\sigma_e)$, we have 
$$u_i(\Gamma_d, \vec{\sigma}, A, \vec{x}) \le
u_i(\Gamma_d, \vec{\sigma}, (\emptyset, \sigma_e), \vec{x})$$
for all $i \not \in T$; no matter what the environment strategy, the
players do at least as well playing $\vec{\sigma}$ as they do by deviating.

Strategy profile  $\vec{\sigma}$ is \emph{$(k,t)$-robust} if it is $t$-immune
and for all $A = (
\vec{\tau}_{K \cup T}, \sigma_e)$ with $|T| \le t$ and $1 \le |K| \le
k$ we have 
$$u_i(\Gamma_d, \vec{\sigma}, A, \vec{x}) \le
u_i(\Gamma_d, \vec{\sigma}, \vec{\tau}_T, \sigma_e, \vec{x})$$
for some $i \in K$.
That is, no matter what strategies the players in $T$ and the
environment are using, at least one player in $K$ cannot gain by deviating. 
A strategy is \emph{strongly} $(k,t)$-robust if
this holds with ``for some $i \in K$'' replaced by ``for all $i \in
K$''.
More intuition for the definition of $(k,t)$-robustness, as well as a
formula definition of $\epsilon$-$(k,t)$-robustness can be found in the
appendix.  
}

Because schedulers can collude with the adversary, 
$t$-immune strategy profiles satisfy an even stronger condition:
deviations by players in a set $T$ with $|T| \le t$ do not make things
worse for the non-deviating players even if the environment
colludes with the players in $T$. 
\begin{proposition}\label{prop:adversary-immunity}
  If $\vec{\sigma}$ is $t$-immune, then for all
sets $T$ of players with $|T| \le t$, strategies $\sigma_e$ and
$\sigma_e'$ for the environment, strategy profiles $\vec{\tau}_T$ for
the players in $T$, input profiles $\vec{x}$
and $\vec{x}'$, and players $i \notin T$, we
have 
\begin{equation}\label{eq:immunity}
u_i(\Gamma_d, (\vec{\sigma}_{-T}, \vec{\tau}_T), \sigma_e', \vec{x}'_T)
\ge u_i(\Gamma_d,
\vec{\sigma}, \sigma_e, \vec{x}_T).
\end{equation}
\end{proposition}
\arxiv{
\begin{proof}
Clearly, if (\ref{eq:immunity}) holds for all $\sigma_e$, 
$\sigma_e'$, $\vec{x}$, and $\vec{x}'$, then $\vec{\sigma}$ is
$t$-immune.  For the converse, suppose by way of contradiction that
$\vec{\sigma}$ is $t$-immunte but
for some $T$, $\vec{\tau}$, $\sigma_e$, $\sigma_e'$, and $i \notin
T$, (\ref{eq:immunity}) does not hold.
Consider a scheduler $\sigma''_e$ that acts
    just like $\sigma_e$, except that if some player $i$ sends a message
to itself it acts 
like $\sigma_e'$.
Then players in $T$ can effectively decrease $i$'s
contradicting that
payoff with scheduler $\sigma_e''$ by sending a message to themselves
and playing as if they had input $\vec{x}'_T$;
that is, there is a strategy $\vec{\tau}'_T$ such that
$$    u_i(\Gamma_d,(\vec{\sigma}_{-T}, \vec{\tau}'_T), \sigma_e'',
\vec{x}_T)
= u_i(\Gamma_d,(\vec{\sigma}_{-T}, \vec{\tau}_T), \sigma_e', \vec{x}'_T)<
u_i(\Gamma_d, \vec{\sigma}, \sigma_e, \vec{x}_T)
= u_i(\Gamma_d, \vec{\sigma}, \sigma_e'', \vec{x}_T),$$
contradicting the assumption that 
$\vec{\sigma}$ is $t$-immune.
The details for the proof are left  
\end{proof}
}

\arxiv{
A similar argument shows that $(k,t)$-robust strategy profiles satisfy a
correspondingly stronger condition,
made precise in the following proposition:}
\disc{$(k,t)$-robust strategy profiles satisfy a
correspondingly stronger condition:}

\begin{proposition}\label{prop:adversary-robustness}
A strategy profile 
  $\vec{\sigma}$ is
$(k,t)$-robust (resp., strongly $(k,t)$-robust) if and only if
it is $t$-immune and 
for all disjoint sets $K$ and $T$ with $1 \le |K| \le k$ and $|T| \le t$,
all strategy profiles $\vec{\tau}_K$, $\vec{\tau}_T$, and
$\vec{\tau}'_T$ for the   
players in $K$ and $T$, respectively, 
all environment strategies $\sigma_e$ and $\sigma'_e$, and all input profiles $\vec{x}$
and $\vec{x}'$,
we have that
\begin{equation}\label{eq:robust}
u_i(\Gamma_d, (\vec{\sigma}_{-(K \cup T)}, \vec{\tau}_K,
\vec{\tau}'_T), \sigma'_e, \vec{x}'_{(K \cup T)}) \le 
u_i(\Gamma_d, (\vec{\sigma}_{-T}, \vec{\tau}_T), \sigma_e, \vec{x}_T)
\end{equation}
for some $i \in K$ (resp., for all $i \in K$). 
\end{proposition}
\arxiv{
\begin{proof}
Again, it is clear that if (\ref{eq:robust}) holds for all $K$ and $T$
with $1 \le |K| \le k$ and $|T| \le t$, all $\vec{\tau}_K$,
$\vec{\tau}_T$, $\vec{\tau}'_T$, $\vec{x}$, and $\vec{x}'$, and some
(resp., all) $i \in K$, then $\vec{\sigma}$ is (strongly)
$(k,t)$-robust.

For the converse, assume by way of contradiction
that $\vec{\sigma}$ is $(k,t)$-robust, but for some disjoint sets $K$
and $T$ with  $1 \le |K| \le k$ and $|T| \le t$, $\vec{\tau}_K$,
$\vec{\tau}_T$, $\vec{\tau}'_T$, $\vec{x}$, and $\vec{x}'$, and all
$i \in K$, (\ref{eq:robust}) does not hold.  
Again, we use the fact that the rational players can effectively
communicate with malicious players and with the 
scheduler.
Consider a
scheduler $\sigma_e''$ that acts like $\sigma_e$ unless some 
player sends a message to itself, in which case it acts like
$\sigma_e'$ and a strategy profile $\vec{\tau}''_T$ in which each
$i \in T$ acts as if it was using strategy $({\tau}_T)_i$, except
that it switches to $({\tau}'_T)_i$ and to pretend have input
$x'_i$ if it receives a message from a rational player asking
it to do so. Then, given input profile $\vec{x}$, strategy profile  
$\vec{\tau}''_T$ for $T$, and scheduler $\sigma_e''$, player $i$ can
gain by sending a message to itself and sending a message to players
in $T$ asking them to follow $\vec{\tau}'_T$ and pretends to have input
$\vec{x}'_T$, and by making players in $K$ play $\vec{\tau}_K$ as if
they had input $\vec{x}'_K$,  rather than
playing $\vec{\sigma}$. This contradicts
the assumption that $\vec{\sigma}$ is $(k,t)$-robust.
The argument for strong $(k,t)$-robustness.
\end{proof}
}
Another property interesting in its own right that follows from
this argument is that $(k,t)$-robust strategies must be
\emph{scheduler-proof}: the expected payoff for all players 
is the same regardless of the scheduler:
\begin{corollary}\label{lemma:scheduler-proof}
    If $\vec{\sigma}$ is $(k,t)$-robust for some $k \ge 1$, then for all
    sets $T$ with $|T| 
\le t$, strategy profiles $\vec{\tau}_T$ for the players in $T$,
environment strategies $\sigma_e$ and $\sigma_e'$, 
input profiles $\vec{x}$,
and players $i \notin T$, we have
$u_i(\Gamma, (\vec{\sigma}_{-T},\vec{\tau}_T),\sigma_e,\vec{x}_T) =
u_i(\Gamma, (\vec{\sigma}_{-T},\vec{\tau}_T),\sigma'_e,\vec{x}_T)$.
\end{corollary}


We have analogous strengthenings of $\epsilon$-$t$-immunity and
$\epsilon$-$(k,t)$-robustness.
\disc{The proof of all these results can be found in the full
  paper~\cite{ADGH18}.} 

\arxiv{
The proofs are essentially identical to that of
    Proposition~\ref{prop:adversary-immunity}, so we omit them here.
    
\begin{proposition}\label{prop:eps-immunity}
If $\epsilon > 0$ and $\vec{\sigma}$ is $\epsilon$-$t$-immune in game $\Gamma$, then
for all sets $T$ of players  with $|T| \le t$, strategy profiles ${\tau}_T$ for the players in
$T$, environment strategies $\sigma_e$ and $\sigma_e'$, 
input profiles $\vec{x}$ 
and $\vec{x}'$, and players $i \notin T$,
we have that 
$$u_i(\Gamma, (\vec{\sigma}_{-T}, \vec{\tau}_T), \sigma_e', \vec{x}'_T)
> u_i(\Gamma, 
\vec{\sigma}, \sigma_e, \vec{x}_T) - \epsilon.$$
\end{proposition}

\begin{proposition}\label{prop:eps-robustness}
A strategy profile 
  $\vec{\sigma}$ is
$\epsilon$-$(k,t)$-robust (resp., strongly $\epsilon$-$(k,t)$-robust)
in game $\Gamma$
if and only if 
it is $\epsilon$-$t$-immune and, 
for all disjoint sets $K$ and $T$ of players with 
$1 \le |K| \le k$ and $|T| \le t$, all strategy profiles $\vec{\tau}_K$ 
and $\vec{\tau}_T$ for players in $K$ and $T$, respectively,
all environment strategies $\sigma_e$ and $\sigma'_e$,
and
all input profiles $\vec{x}$ 
and $\vec{x}'$,
we have
that $$u_i(\Gamma, (\vec{\sigma}_{-(K \cup T)}, \vec{\tau}_K,
\vec{\tau}_T), \sigma_e', \vec{x}_{(T\cup K)}') <
u_i(\Gamma, (\vec{\sigma}_{-T}, \vec{\tau}_T), \sigma_e, \vec{x}_{T})
+ \epsilon$$ 
for some $i \in K$ (resp., for all $i \in K$). 
\end{proposition}

It will be useful for our later results that we can actually improve
on the bound of  
$\epsilon$ in Propositions~\ref{prop:eps-immunity}
and~\ref{prop:eps-robustness}.

\begin{proposition}\label{prop:compactness-immunity}
If $\vec{\sigma}$ is an $\epsilon$-$t$-immune strategy in a finite
game $\Gamma$, then
there exists $\epsilon_0$ with $0 < \epsilon_0 < \epsilon$ such
that for all
sets of players $T$ with $|T| \le t$, strategy profiles $\vec{\tau}_T$
for the players in $T$, environment strategies $\sigma_e$ and 
$\sigma'_e$, input profiles $\vec{x}$
and $\vec{x}'$,
and players $i \notin T$,
we have that
$$u_i(\Gamma, (\vec{\sigma}_{-T}, \vec{\tau}_T), \sigma_e',
\vec{x}'_T) > u_i(\Gamma,  
\vec{\sigma}, \sigma_e, \vec{x}_T) - \epsilon_0.$$ 
\end{proposition}
\arxiv{
\begin{proof}
  Since, by Proposition~\ref{prop:eps-immunity},
  for each choice of $\vec{\tau}_T$, $\sigma_e$, and $\sigma_e'$,
  we have 
$$u_i(\Gamma, \vec{\sigma}, \sigma_e, \vec{x}_T) - u_i(\Gamma,
(\vec{\sigma}_{-T}, \vec{\tau}_T), \sigma_e', \vec{x}'_T) < \epsilon,$$
%
and the space of player strategy profiles, environment strategies, and
input value profiles is compact, if we take the sup of the left-hand
side over all choices of strategy profiles $\vec{\tau}_T$, 
environment strategies $\sigma_e$ and $\sigma_e'$, and input profiles
$\vec{x}$
and $\vec{x}'$, it takes on some maximum value $\epsilon_1 < \epsilon_0$.
We can then take $\epsilon_0 = (\epsilon + \epsilon_1)/2$.
\commentout{
Since the set of actions, inputs and messages is finite, and the
mediator game is in canonical form, the set of possible local history
profiles is finite.  
$\{O(\vec{\sigma}, A, \vec{x})\}_{A, \vec{x}}$ is compact. In
particular, since by Proposition~\ref{prop:eps-immunity} we have that
$u_i(\Gamma, \vec{\sigma}, A, \vec{x}) > u_i(\Gamma, \vec{\sigma},
\sigma'_e, \vec{x}) - \epsilon$, it follows that 
$$s:= \sup_{A, \vec{x}}\{u_i(\Gamma_d, \vec{\sigma}, \sigma'_e,
\vec{x}) - u_i(\Gamma_d, \vec{\sigma}, \sigma'_e, \vec{x})\} <
\epsilon$$ 

The desired result holds by picking $\epsilon_0 := (s + \epsilon)/2$.
}
\end{proof}

Using Proposition~\ref{prop:eps-robustness}, we get a similar result
for $\epsilon$-$(k,t)$-robustness. The proof is analogous to that of
Proposition~\ref{prop:compactness-immunity}. 
}

\begin{proposition}\label{prop:compactness-robustness}
If $\Gamma$ is a finite game and $\vec{\sigma}$ is 
a (strongly) $\epsilon$- $(k,t)$-robust strategy 
in $\Gamma_d$, then
there exists $\epsilon_0$ with $0  < \epsilon_0 < \epsilon$ such that
for all disjoint sets $K$ and $T$ of players with 
$1 \le |K| \le k$ and $|T| \le t$, all strategy profiles $\vec{\tau}_K$ 
and $\vec{\tau}_T$ for players in $K$ and $T$, respectively,
all environment strategies $\sigma_e$ and $\sigma'_e$, and all input
profiles $\vec{x}$
and $\vec{x}'$, 
we have that $$u_i(\Gamma, (\vec{\sigma}_{-(K \cup T)},
\vec{\tau}_K, \vec{\tau}_T), \sigma_e, \vec{x}'_{(K \cup T)}) <
u_i(\Gamma, (\vec{\sigma}_{-T}, \vec{\tau}_T), \sigma'_e, \vec{x}_T) +
\epsilon_0$$
for some $i \in K$ (resp., all $i \in K$).
\end{proposition}

\fullv{
\begin{proposition}\label{prop:eps-immunity}
If $\epsilon > 0$ and $\vec{\sigma}$ is $\epsilon$-$t$-immune, then
  for all $A = (\vec{\tau}_T, \sigma_e)$, 
all environment strategies $\sigma'_e$, and all and input profiles $\vec{x}$ we
have that 
$$u_i(\Gamma_d, \vec{\sigma}, A, \vec{x}) \ge u_i(\Gamma_d,
\vec{\sigma}, \sigma'_e, \vec{x}) - \epsilon.$$
\end{proposition}

\begin{proposition}\label{prop:eps-robustness}
A strategy profile 
  $\vec{\sigma}$ is
$\epsilon$-$(k,t)$-robust (resp., strongly $\epsilon$-$(k,t)$-robust) if and only if
it is $\epsilon$-$t$-immune and 
for all $A = (
\vec{\tau}_K, \vec{\tau}_T, \sigma_e)$ with $1 \le |K| \le k$ and $|T| \le
t$, 
all strategy profiles $\vec{\tau}_T'$ for players in $T$, 
all environment strategies $\sigma'_e$, and all input profiles $\vec{x}$,
we have
that $$u_i(\Gamma_d, \vec{\sigma}, A, \vec{x}) \le
u_i(\Gamma_d, (\vec{\sigma}_{-T}, \vec{\tau}_T'), \sigma'_e, \vec{x}) + \epsilon$$
for some $i \in K$ (resp., for all $i \in K$). 
\end{proposition}
}

}
\subsection{Proof of
  Theorem~\ref{thm:upperbound-default-no-punish}}\label{sec:thm4.1proof} 


By Theorem~\ref{thm:errorless-simulation}, if $n
> 4k+4t$, there exists a strategy 
$\vec{\sigma}_{\ACT}$ that $(k+t)$-bisimulates and $(k+t)$-emulates
$\vec{\sigma}+\sigma_d$,
in which the expected number of messages is $O(nNc)$.  It is 
immediate 
from the definition of $(k+t)$-bisimulation
that $\vec{\sigma}_{\ACT}$ implements $\vec{\sigma} + \sigma_d$. Since the probability of deadlock is 0, the action that players play
in case of deadlock are irrelevant, so this approach works equally
well for the AH approach and the default-move approach.
It remains to show that, for each utility variant $\Gamma_d(\vec{u}')$
of $\Gamma_d$, if $\vec{\sigma}+\sigma_d$ is a (strongly) $(k,t)$-robust equilibrium
in $\Gamma_d(\vec{u}')$, then $\vec{\sigma}_{\ACT}$ is a (strongly)
$(k,t)$-robust equilibrium in $\Gamma_{\ACT}(\vec{u}')$. We start by showing that $\vec{\sigma}_{\ACT}$ is $t$-resilient in $\Gamma_{\ACT}(\vec{u}')$.

Given $T$ with $|T| \le t$, $\vec{\tau}_T$, and $\sigma_e$, by
Theorem~\ref{thm:errorless-simulation}, there exists a function $H$  
from strategies to strategies and
$\sigma'_e$ such 
that, for all input profiles $\vec{x}$, we have
$$u_i'(\Gamma_{\ACT}(\vec{u}'),
((\vec{\sigma}_{\ACT})_{-T}, \tau_T), \sigma_e, \vec{x}) =
u_i'(\Gamma_d(\vec{u}'), 
(\vec{\sigma}_{-T}, H(\vec{\tau}_T)), \sigma_e', \vec{x})$$ for all players $i$. 
There also exists a scheduler of the form $\sigma''_e$ such
that $$u_i'(\Gamma_{\ACT}(\vec{u}'), \vec{\sigma}_{\ACT},
\sigma'_e, \vec{x}) = u_i'(\Gamma_d(\vec{u}'), \vec{\sigma},
\sigma''_e, \vec{x}).$$

Since $\vec{\sigma}$ is $t$-immune,  for all $i \notin T$ we have
$$\begin{array}{lll}  
    & u_i'(\Gamma_{\ACT}(\vec{u}'), ((\vec{\sigma}_{\ACT})_{-T},
  \vec{\tau}_T), \sigma_e, \vec{x}_T) & \\
= & u_i'(\Gamma_d(\vec{u}'), (\vec{\sigma}_{-T}, H(\vec{\tau}_T)), \sigma_e', \vec{x}_T) & \\
\ge & u_i'(\Gamma_d(\vec{u}'), \vec{\sigma}, \sigma''_e,
\vec{x}_T) & \mbox{[by Lemma~\ref{prop:adversary-immunity}]} \\ 
= & u_i'(\Gamma_{\ACT}(\vec{u}'), \vec{\sigma}_{\ACT},
\sigma'_e, \vec{x}_T).
\end{array}$$ 
%
Therefore, $\vec{\sigma}_{\ACT}$ is $t$-immune.

To show (strong) $(k,t)$-robustness, taking
$\vec{\tau}_T$, $\sigma_e$, and $\sigma_e'$ as above,
suppose that $K$ is a set of players disjoint from $T$ such that $|K|
\le k$, and the players in $K$ play $\vec{\tau}_K$.  
Then, by Theorem~\ref{thm:errorless-simulation},
there exists
$\sigma^*_e$ such
that $$u_i'(\Gamma_{\ACT}(\vec{u}'), ((\vec{\sigma}_{\ACT})_{-(K \cup
    T)}, \vec{\tau}_K, \vec{\tau}_T), \sigma_e,
\vec{x}) = u_i'(\Gamma_{d}(\vec{u}'), (\vec{\sigma}_{-(K\cup T)},
H(\vec{\tau}_K), H(\vec{\tau}_T)), \sigma_e^*,
\vec{x}_{(K \cup T)})$$ for all players $i$.
\arxiv{By Corollary~\ref{prop:adversary-robustness}, if $\vec{\sigma}
  + \sigma_d$ is  $(k,t)$-robust (resp., strongly $(k,t)$-robust)  in
  $\Gamma_d(\vec{u}')$, then}
\disc{By Corollary~\ref{prop:adversary-robustness}, if $\vec{\sigma}
  + \sigma_d$ is  (strongly) $(k,t)$-robust in
  $\Gamma_d(\vec{u}')$, then}
$$u_i'(\Gamma_{d}(\vec{u}'), (\vec{\sigma}_{-(K \cup T)},
H(\vec{\tau}_K), H(\vec{\tau}_T)), \sigma_e^*, \vec{x}_{(K \cup T)}) \le 
u_i'(\Gamma_{d}(\vec{u}'), (\vec{\sigma}_{-T}, H(\vec{\tau}_T)),
\sigma_e', \vec{x}_T)$$ for some 
(resp., all) $i \in K$.  For those $i \in K$ for which this inequality
holds, we have
$$\begin{array}{lll} 
  &u_i'(\Gamma_{\ACT}(\vec{u}'), ((\vec{\sigma}_{\ACT})_{-(K \cup
        T)}, \tau_K, \tau_T), \sigma_e, \vec{x}_{(K \cup T)})\\
  = &u_i'(\Gamma_{d}(\vec{u}'), (\vec{\sigma}_{-(K\cup T)},
  H(\vec{\tau}_K), H(\vec{\tau}_T)), \sigma_e^*, \vec{x}_{(K\cup T)})\\ 
= &u_i'(\Gamma_{d}(\vec{u}'), (\vec{\sigma}_{-T},
H(\vec{\tau}_T)), \sigma_e', \vec{x}_T)\\
= & u_i'(\Gamma_{\ACT}(\vec{u}'), ((\vec{\sigma}_{\ACT})_{-T}, \sigma_e,
\vec{x}_T). \end{array}$$ 
It follows that $\vec{\sigma}_{\ACT}$ is (strongly) $(k,t)$-robust in
$\Gamma_{\ACT}(\vec{u}')$.  

\disc{The proof of Theorem~\ref{thm:upperbound-default-no-punish1} 
is essentially the same as that of
Theorem~\ref{thm:upperbound-default-no-punish}, except that we now use
Theorem~\ref{thm:error-simulation} instead of
Theorem~\ref{thm:errorless-simulation}.
Details can be found in the full paper.}
\arxiv{
\subsection{Proof of Theorem~\ref{thm:upperbound-default-no-punish1}}

The argument is essentially the same as that used for
Theorem~\ref{thm:upperbound-default-no-punish}, except that we now use
Theorem~\ref{thm:error-simulation} instead of
Theorem~\ref{thm:errorless-simulation}.
By Theorem~\ref{thm:error-simulation}, for all $\epsilon' \in (0,1]$,
  there exists a protocol $\vec{\sigma}_{\ACT}$ that
    $\epsilon$-$(k+t)$-bisimulates $\vec{\sigma}$, 
$\epsilon$-$(k+t)$-emulates $\vec{\sigma}$,  and uses $O(nNc)$ messages in expectation.
  It follows that $\vec{\sigma}_{\ACT}$
  $\epsilon'$-implements $\vec{\sigma}$ and has at most a
  probability $\epsilon'$ of deadlock. We show next that we can make
  $\epsilon'$   sufficiently small so that it becomes irrelevant if
    we work   with the AH approach or the default-move approach.


We can now prove $\epsilon$-$(k,t)$-robustness.  
Suppose that $\vec{\sigma}+\sigma_d$ is 
a (strongly) 
$\epsilon$-$(k,t)$-robust equilibrium in
the utility variant $\Gamma_d(\vec{u}')$ of $\Gamma_d$. We 
show that $\vec{\sigma}_{\ACT}$ is $\epsilon$-$t$-immune in
$\Gamma_{\ACT}(\vec{u}')$.
Since $\vec{\sigma}_{\ACT}$ $\epsilon'$-$(k+t)$ emulates $\vec{\sigma}
+ \sigma_d$, 
given a set $T$ of players with $|T| \le t$, a strategy profile
$\vec{\tau}_T$ for the players in $T$, and an environment strategy
$\sigma_e'$ in the cheap-talk game, by
Theorem~\ref{thm:error-simulation}, there exists a function $H$
as in the definition of $t$-emulation. 
Moreover, for all inputs $\vec{x}$, 
we can associate histories in the mediator game and
histories in the cheap-talk game in such a way that the set of
histories where the outcomes differ has probability at most $\epsilon'$.
Since all utilities are in the range $[-M/2,M/2]$, by assumption,
the maximum difference in utility between two outcomes in
the underlying game is $M$.  Thus,
there exists
an environment strategy $\sigma_e'$
such that for all input profiles $\vec{x}$, we have 
$$u_i'(\Gamma_{\ACT}(\vec{u}'), ((\vec{\sigma}_{\ACT})_{-T}, H(\vec{\tau}_T)),
\sigma_e', \vec{x}) >
u_i'(\Gamma_d(\vec{u}'), (\vec{\sigma}_{-T}, \vec{\tau}_T),
\sigma_e, \vec{x}) - \epsilon'M$$ 
%
for all $i \not \in T$.
Theorem~\ref{thm:error-simulation} also guarantees 
that there exists an environment strategy $\sigma''_e$ such
that $$u_i'(\Gamma_{\ACT}(\vec{u}'), \vec{\sigma}_{\ACT}, \sigma_e',
\vec{x}) < u_i'(\Gamma_d(\vec{u}'), \vec{\sigma}, \sigma''_e,
%
%
\vec{x}) + \epsilon'M.$$ 
Since $\vec{\sigma}$ is $\epsilon$-$t$ immune in $\Gamma_d(\vec{u}')$,
by Proposition~\ref{prop:compactness-immunity}, 
there exists a value $\epsilon_0$ with $0 < \epsilon_0 < \epsilon$ such that
$$\begin{array}{lll}
  &u_i'(\Gamma_{\ACT}(\vec{u}'), ((\vec{\sigma}_{\ACT})_{-T}, \vec{\tau}'_T),
  \sigma_e', \vec{x}_T) & \\ 
    > &u_i'(\Gamma_d(\vec{u}'), (\vec{\sigma}_{-T}, H(\vec{\tau}_T)),
\sigma_e, \vec{x}_T) - \epsilon'M & \\ 
> & u_i'(\Gamma_d(\vec{u}'), \vec{\sigma}, \sigma''_e, \vec{x}_T) -
\epsilon_0 - \epsilon'M&
\\ 
> & u_i'(\Gamma_{\ACT}(\vec{u}'), \vec{\sigma}_{\ACT}, \sigma_e',
\vec{x}_T) - \epsilon_0 - 2\epsilon'M.& \end{array}$$ 
If we take $\epsilon' = (\epsilon - \epsilon_0)/2M$, this shows that
$\vec{\sigma}_{\ACT}$ is $(\epsilon,t)$-immune with both the
AH approach and the default-move approach.

To show
(strong) $\epsilon$-$(k,t)$-robustness,
keeping $T$, $\vec{\tau}_T$, $H$, $\sigma_e$, and $\sigma_e'$ as above,
for all sets $K$ of players disjoint from $T$ with 
$1 \le |K| < k$ and 
strategy profiles $\vec{\tau}_K$, 
there exists an environment strategy $\sigma^*_e$
and a value $\epsilon_0$ with $0 < \epsilon_0 < \epsilon$ 
 such
that for all input profiles $\vec{x}$,
%
$$\begin{array}{lll}
  &  u_i'(\Gamma_{\ACT}(\vec{u}'), ((\vec{\sigma}_{\ACT})_{-(K \cup
    T)}), \vec{\tau}_K,\vec{\tau}_T), \sigma_e', \vec{x}_{(K \cup T)}) & \\
< & u_i'(\Gamma_{d}(\vec{u}'), (\vec{\sigma}_{-(K \cup T)},
H(\vec{\tau}_K), H(\vec{\tau}_T)), \sigma_e, \vec{x}_{(K \cup T)}) + \epsilon'M & \\
< & u_i'(\Gamma_{d}(\vec{u}'), (\vec{\sigma}_{-T}, H(\vec{\tau}_T)),
\sigma_e'', \vec{x}_T) + \epsilon_0 + \epsilon'M
& \mbox{[by Proposition~\ref{prop:compactness-robustness}]} \\ 
< & u_i'(\Gamma_{\ACT}(\vec{u}'), (\vec{\sigma}_{\ACT})_{-T},
\vec{\tau}), \sigma_e, \vec{x}_T) + \epsilon_0 +
2\epsilon'M & 
\mbox{[by  Theorem~\ref{thm:error-simulation}].} \end{array}$$ 

for some (for all) $i \in K$.
This shows that if we take
$\epsilon' := (\epsilon - \epsilon_0)/2M$,  
then
$\vec{\sigma}_{\ACT}$ is $\epsilon$-$(k,t)$-robust. Note that this
argument works for both the AH approach and the default-move approach
since it does not depend on the actions played in deadlock. 
}
\subsection{Proof of Theorem~\ref{thm:punish}}\label{sec:punish}
\commentout{
We cannot use quite the same
approach to prove Theorem~\ref{thm:punish} as that used to prove
Theorem~\ref{thm:upperbound-default-no-punish}.
The main problem is that 
Theorem~\ref{thm:errorless-simulation} guarantees $t$-emulation can
only if $t < n/4$.
However, given a 
$(k+t)$-punishment
strategy, we can force sufficient cooperation from the rational
players for our purposes.
}

The proof of Theorem~\ref{thm:punish} is similar in spirit to that of
Theorem~\ref{thm:upperbound-default-no-punish}.  The main problem we 
have to deal with is that of ensuring that rational players
participate.  To force participation, we have the honest players put
the punishment strategy in their ``wills'', so that if
$\vec{\sigma}_{\ACT}$ ends in deadlock, the rational players will be 
punished.   Unfortunately, a naive implementation of this approach
does not work, as the following example shows.

Consider an underlying game $\Gamma$ for $n > 3k$ players where the
set of actions is $A := \{0,1,\bot\}$.
If at least $k+1$ players play $\bot$, all players get a payoff of
1.1; if $k$ or fewer players play $\bot$ and all players play
either 0 or $\bot$, then all players get a payoff of 1; if $k$ or
fewer players play $\bot$ and all players play either $\bot$ or 1, then all
players get a payoff of 2; otherwise, all players get 0.  
Let $\Gamma_d$ be an extension of $\Gamma$ with a mediator.  Suppose that
the mediator $d$ uses the following strategy:
First, $d$ chooses a value $b \in \{0,1\}$ with equal probability.
Then $d$ chooses
  $a \in \{0,1\}$ with equal probability and sends the message
  $a+bi \pmod 2$ to player $i$ (note that the same $a$ is used in all
these messages).  
Finally, $d$ sends the message ``output $b$; STOP'' to all players (so
the strategy is in canonical form).

Let $\sigma_i$ be the strategy where 
player $i$ ignores the message $a+bi$ and plays $b$ after
receiving the message ``output $b$''.  It is easy to check that
$\vec{\sigma}$ is a $k$-resilient equilibrium in the mediator game,
and gives players an 
expected payoff of 1.5.  Moreover, playing $\bot$ is a $k$-punishment
strategy with respect to $\vec{\sigma}$, since if all but $k$ players
play $\bot$, then everyone gets a payoff of 1.1 (since at least $k+1$
players play $\bot$), which is less than 1.5.

The naive approach to implementing the mediator does not work for
this game, at least with the punishment strategy $\bot$.  
For example, suppose that after receiving the
messages $a + bi \pmod 2$, the rational players communicate with each
other.  Moreover, suppose that the set $K$ of rational players
includes $i$ and $j$ such that $i-j$ is odd.  Then the rational
players can compute $b$.  If $b=0$, they actually prefer their payoff
with the punishment strategy to their payoff with 
$\vec{\sigma}_{\ACT}$.  Thus, they will stop sending messages.  The
simulation will not terminate, so the punishment strategy in the
players' wills will be applied, making the rational players better
off.  Thus, we cannot
simulate the mediator with this approach.
%
Of course, there are punishment strategies in this game that would
lead to cooperation (e.g., randomizing between 0 and 1). Nevertheless,
this example shows that using an arbitrary punishment strategy may not suffice
to force the rational players to cooperate.

The problem here is that the mediator tells each player $i$ what 
$a+bi$ is.   We do not want the mediator to send such
unnecessary information.  But what counts as unnecessary?  
\arxiv{As we now show, }
\arxiv{
for each strategy profile $\vec{\sigma} + \sigma_d$ of a mediator game
we can construct a strategy $\vec{\sigma}^m + \sigma_d^m$ that
implements $\vec{\sigma} + \sigma_d$ and leaks no information. More
precisely, there exists a function $\mifunction$ from strategy
profiles to strategy profiles such that, for all strategy profiles
$\vec{\sigma} + \sigma_d$, $\mifunction(\vec{\sigma} + \sigma_d)$
implements  $\vec{\sigma} + \sigma_d$ and essentially all 
the mediator sends each player when playing $\mifunction(\vec{\sigma}
+ \sigma_d)$ is the action to play in the underlying 
game.  (If we require only weak implementation, then this is exactly the
case; for implementation, the messages can also include a round
number.)
Moreover, if $\vec{\sigma}+\sigma_d$ is $(k,t)$-robust (resp., 
strongly $(k,t)$-robust, $\epsilon$-(k,t)-robust, strongly
$\epsilon$-(k,t)-robust), then
is $\mifunction(\vec{\sigma} + \sigma_d)$.  
The construction of $\mifunction$ proceeds as follows:
}
\disc{As we show in the full paper, 
for all strategy profiles $\vec{\sigma}+\sigma_d$ for a mediator
game, there exists a 
strategy profile
$\vec{\sigma}^m +\sigma_d^m$ that implements $\vec{\sigma}+
\sigma_d$
and can be viewed as \emph{minimally informative} in that it leaks no
information up to the last step.
More precisely,  
when using strategy $\sigma_d^m$, essentially all that 
the mediator sends each player
is the action to play in the underlying
game.  If we require only weak implementation, then this is exactly the
case;
to get full implementation, 
the messages sent by the mediator also include a round 
number.
Moreover, if $\vec{\sigma}+\sigma_d$ is $(k,t)$-robust (resp., 
strongly $(k,t)$-robust, $\epsilon$-(k,t)-robust, strongly
$\epsilon$-(k,t)-robust), then so is $\vec{\sigma}^m+\sigma_d^m$.
To get weak implementation, $\vec{\sigma}^m+\sigma_d^m$ requires only
$O(n)$ messages, but for implementation, our construction requires
$\msgbound$ messages, which we conjecture is
necessary.
}

\commentout{
\arxiv{
Implementation is straightforward from its construction, the details
about robustness and number of messages are shown in the full
paper~\cite{fullPaper_}.}} 
\arxiv{
\commentout{
\begin{definition}\label{def:minimally-informative}
The strategy profile $\vec{\sigma} +
\sigma_d$ in a mediator game $\Gamma_d$ is \emph{minimally informative} if the following
conditions hold:
\begin{itemize}
  \item There exists a value $R$ such that the mediator sends all the
  players exactly $R$ messages. 
\item For all $r \le R$, the $r$th message from the mediator to each
  player is sent   simultaneously. If $r < R$,
  the $r$th message sent by the
  mediator to each player is 
  the same and, without loss of generality, is just $r$.
  The $R$th message from the mediator to each player consists of the
  STOP instruction and an action in the underlying game.
\end{itemize}
\end{definition}

\begin{lemma}\label{lem:minimallyinformative} 
Given a finite mediator game $\Gamma_d$ that extends an underlying
game $\Gamma$, a canonical strategy profile
$\vec{\sigma}$ for the
players, and
a strategy $\sigma_d$ for the 
mediator in $\Gamma_d$, there
exists a minimally informative canonical strategy $\vec{\sigma}'$ for
the players and 
$\sigma_d'$ for the 
mediator 
such that 
$(\vec{\sigma}', \sigma_d')$
implements (resp., weakly implements) $(\vec{\sigma}, \sigma_d)$,
and if $\vec{\sigma}+\sigma_d$ is $(k,t)$-robust (resp.,
strongly $(k,t)$-robust, $\epsilon$-(k,t)-robust, strongly $\epsilon$-(k,t)-robust), then so is
$\vec{\sigma}+\sigma_d$.
Moreover,
the number of messages sent in each history of  
$\vec{\sigma}_{\ACT}$ is 
\msgbound
(resp., $O(n)$).
\end{lemma}
}
\commentout{
\begin{proof}
} 
  Let    $\mathcal{S}^{det}_{\Gamma_d,e}/\!\!\sim$ denote the set of
  $\sim$-equivalence classes.  Thus,   
  $\mathcal{S}^{det}_{\Gamma_d,e}/\!\!\sim$ is essentially the
    set of deterministic environment strategies that result in
    different outcomes in the 
    underlying game when the players and the
mediator use $\vec{\sigma} + \sigma_d$ in $\Gamma_d$.  

The intuition underlying 
$\mifunction(\sigma_d)$
 is that the mediator chooses an
equivalence class in $\mathcal{S}^{det}_{\Gamma_d,e}/\!\!\sim$, chooses
an environment strategy $\sigma_e$ in the equivalence 
class, and simulates the outcome of $(\vec{\sigma} + \sigma_d,
\sigma_e)$.  In order to get an implementation, we must ensure that it
is possible for 
$\mifunction(\sigma_d)$ 
to choose all possible equivalences
classes in $\mathcal{S}^{det}_{\Gamma_d,e}/\!\!\sim$.  
To do this, let $R$ be the least integer such that 
$(Rn)! \ge
|\mathcal{S}^{det}_{\Gamma_d,e}/\!\!\sim\!\!|$.  We show below that we
can take $R = 2^{rn\log(n)}$.
    The mediator $\sigma_d$ sends $R$ messages 
    to each player.  As we
    shall see, this suffices for the mediator to choose all possible
    equivalence classes in $\mathcal{S}^{\det}_{\Gamma_d,e}/\!\!\sim$.  

    The strategy
      $\mifunction(\sigma_i)$
       is straightforward:
player $i$ initially sends the mediator the message $(i,0,x_i)$, where
$x_i$ is $i$'s input.
Then for 
$1 \le r < R$, after receiving a message with content $r$ from the
mediator, player $i$ sends the mediator 
$(i,r,x_i)$. When $i$ receives a message
of the form (STOP, $a_i$) from the mediator, $i$ plays $a_i$ and
halts.

The mediator 
$\mifunction(\sigma_d)$
 proceeds as follows: it sends each player $i$
$R-1$
 messages, where the $r$th message just says ``$r$''.
 It then waits until
there are at least $n-k-t$ players
from which it has received a \emph{valid} and \emph{complete} set of
messages,
where a set of messages from a player $i$ is valid and complete if
for all $r$ with $0 \le r \le R-1$,
the mediator has received exactly one message of the form $(i,r,x)$,
  where $x$ is an input value that $i$ could have,  and all the $x$
  values are the same in these $R$ messages.
  The next time that the mediator is scheduled, it sends a STOP message to each
  player with an action to perform.  We next explain how the mediator
   calculates which actions the players perform.   
   
Let $P$ be the set of players from whom the mediator has
received a valid and complete of messages when it is next scheduled.
There are two cases. 
If $P$ consists of all players, this means that the mediator has received
$Rn$ messages.  There are $(Rn)!$ orders that these messages could
have come in.
Moreover, for each possible order of messages, there 
is a deterministic scheduler $\sigma_e'$
that delivers the messages in just this order.
By choice of $R$, $(Rn)! \ge |\mathcal{S}^{det}_{\Gamma_d,e}/\!\!\sim'|$,
so there is a surjective
mapping 
$H_P$ from each message order to a scheduler $\sigma_e$
in the game $\Gamma_d$.
The mediator then simulates a computation of
$(\vec{\sigma},\sigma_d)$ with the scheduler $\sigma_e$ corresponding to
the message order it actually received (generating the
randomness for all the players and for $\sigma_d$)
using the input $\vec{x}$ that it received from the players,
and sends each
player $i$ the action that results from this simulation as its $R$th
message.  

Now suppose $P$, the set of players from whom the mediator has
received a valid and complete of messages when it is next scheduled,
does not consist of all players. Let $\sigma_P'$ be a scheduler in the
game $\Gamma_d'$ that resulted in this message order. 
Consider a fixed scheduler
$\sigma_P$ in the game $\Gamma_d$ where the messages of all players
not in $P$ are delayed until after the mediator has sent all STOP
messages.  (There must be such a scheduler, since $|P| \ge n-k-t$).
Let $\vec{x}_P$ be the profile of inputs that the mediator has
received from the players in $P$, and extend it arbitrarily to an
input profile $\vec{x}$.  The mediator 
$\mifunction(\sigma_d)$
 simulates a computation of
$(\vec{\sigma},\sigma_d)$ with the scheduler $\sigma_P$ 
(again, generating the 
randomness for all the players and for $\sigma_d$)
using some profile $\vec{x}$ that extends the profile $\vec{x}_P$ it
received from the players in $P$ (it doesn't matter which vector
$\vec{x}$ is used, since the mediator $\sigma_d$ does not receive
messages from players not in $P$), and sends each
player $i$ the action that results from this simulation as its $R$th
message.
\commentout{
 As above, on all input profiles, 
$(\mifunction(\vec{\sigma} + \sigma_d),\sigma_P')$
and $(\vec{\sigma}+\sigma_d,\sigma_P)$ induce
  the same distribution in $\Delta(A)$.  It follows that
  $\vec{\sigma}' + \sigma_d'$ implements $\vec{\sigma} + \sigma_d$.
  }

If we require only weak implementation, we can make do with far fewer
messages.  Each player $i$ just sends the mediator 
$\mifunction(\sigma_d)$
an
initial message of the form $(i,x_i)$.
The mediator waits until it has these initial messages from $n-k-t$
players.  
If the mediator has received
messages from the players in $P$ when it is next scheduled, it
simulates the mediator $\sigma_P$ on some input $\vec{x}$ that extends
$\vec{P}$.  

We now show that this construction has the required properties.

\begin{lemma}\label{lem:correctness-of-construction}
Given a finite mediator game $\Gamma_d$ that extends an underlying
game $\Gamma$, 
a canonical strategy profile
$\vec{\sigma}$ for the
players, and
a strategy $\sigma_d$ for the 
mediator in $\Gamma_d$,
then
$\mifunction(\vec{\sigma} + \sigma_d)$
implements (resp., weakly implements) $(\vec{\sigma}, \sigma_d)$,
and if $\vec{\sigma}+\sigma_d$ is $(k,t)$-robust (resp.,
strongly $(k,t)$-robust, $\epsilon$-(k,t)-robust, strongly $\epsilon$-(k,t)-robust), then so is
$\mifunction(\vec{\sigma}+\sigma_d)$.
Moreover,
the number of messages sent in each history of  
$\mifunction(\vec{\sigma}+\sigma_d)$ is 
\msgbound
(resp., $O(n)$).
\end{lemma}

\begin{proof}
  Implementation (resp. weakly implementation) follows from the construction. 
  Suppose that $\vec{\sigma}+\sigma_d$ is $\epsilon$-$(k,t)$-robust
(resp., strongly $\epsilon$-$(k,t)$-robust, $(k,t)$-robust, strongly
$(k,t)$-robust).  We must show that 
$\mifunction(\vec{\sigma} + \sigma_d)$
has the
corresponding property.  The proofs are essentially the same in all
cases; moreover, the same argument works for both constructions
(i.e., both the one the gives implementation and the one that gives
weak implementation). For definiteness, we deal with
$\epsilon$-$(k,t)$-robustness here.

We start by showing $\epsilon$-$t$-immunity.  If
$\mifunction(\vec{\sigma}+ \sigma_d)$
is not $\epsilon$-$t$-immune, then there must be 
scheduler $\sigma_e'$, a set $T$ of players with $|T| \le t$,
a strategy $\vec{\tau}_T$ for the players in $T$, and an input profile
$\vec{x}_T$ that shows this.
That is, for some $i \notin T$, we have 
$$u_i(\Gamma_d', (
\mifunction(\vec{\sigma}_T)
,\tau_{T}),\sigma_e',\vec{x}_T) \le
u_i(\Gamma_d', 
\mifunction(\vec{\sigma})
,\sigma_e',\vec{x}_T) - \epsilon.$$
We can assume without loss of generality that $\sigma_e'$ is deterministic.
Moreover, we can further assume without loss of generality that, when
playing with $\vec{\tau}$, the malicious players deviate only
by sending an input other than their actual input.
All other deviations correspond to playing with a different
(deterministic) scheduler.  Specifically,
since the mediator ignores all messages that are not sent by a player
in the set $P$, the outcome is equivalent to playing with a scheduler
that delays these messages until after the mediator sends its $R$th
message.  And if a deviating player $j$ does not send a message
that it should send or sends a message late, we can just consider the
scheduler that sends messages in the order that 
$\mifunction(\sigma_d)$
 actually
received them.  
It now follows that 
there exists a scheduler 
$\sigma''_e$
and an input profile $\vec{x}'_T$ 
such that $$u_i(\Gamma_d', (
\mifunction(\vec{\sigma})_T
,\tau_T),
\sigma'_e, \vec{x}_T) = u_i(\Gamma_d', 
\mifunction(\vec{\sigma})
, 
\sigma''_e, \vec{x}_T').$$
Also, by construction, 
there exists a scheduler $\sigma_e$ such that $$u_i(\Gamma_d',
\mifunction(\vec{\sigma})
, \sigma''_e, \vec{x}_T') = u_i(\Gamma_d,
\vec{\sigma}, \sigma_e, \vec{x}_T')$$ for some scheduler
$\sigma_e$. 
Similarly, there exists a scheduler $\sigma_e^*$ such that
$$u_i(\Gamma_d',
\mifunction(\vec{\sigma})
, \sigma'_e, \vec{x}_T') = u_i(\Gamma_d,
\vec{\sigma}, \sigma_e^*, \vec{x}_T').$$
Thus, we have
$$u_i(\Gamma_d, \vec{\sigma}, \sigma_e, \vec{x}_T) 
\le u_i(\Gamma_d, \vec{\sigma}, \sigma_e^*, \vec{x}_T') 
-
 \epsilon.$$
By Proposition~\ref{prop:eps-immunity}, this contradicts the
assumption that $\vec{\sigma}$ is $\epsilon$-$t$-immune.

Now let $K$ be any subset disjoint from $T$ such that $1 \le |K| \le
k$. Then, using an analogous argument, there exists a scheduler
$\sigma_e'''$ and an input profile $\vec{x}_T''$ such
that $$u_i(\Gamma_d',
\mifunction(\vec{\sigma})_{-(T \cup K)}
,
\sigma'_e, \vec{x}_{(T\cup K)}) = u_i(\Gamma_d', 
\mifunction(\vec{\sigma})
,
\sigma'''_e, \vec{x}_{(T \cup K)}''),$$
for all players $i$.
There also exists a scheduler $\sigma_e^{**}$ and an input profile
$\vec{x}'''_{T}$ such that
$$u_i(\Gamma_d, 
\mifunction(\vec{\sigma})_{-T}
, \vec{\tau}_T,
\sigma'_e, \vec{x}_{(T \cup K)}) = u_i(\Gamma_d, 
\mifunction(\vec{\sigma})
,
\sigma_e^{**}, \vec{x}_{(T  \cup K)}''')$$ 
for all players $i$.
Thus,
$$\begin{array}{lll}
  & u_i(\Gamma_d',
  \mifunction(\vec{\sigma})_{-(T \cup K)}
  , \sigma'_e,
  \vec{x}_{(T\cup K)}) & \\ 
= & u_i(\Gamma_d, 
\mifunction(\vec{\sigma})
, \sigma'''_e, \vec{x}_{(T
\cup K)}'') & \\ 
< &u_i(\Gamma_d, 
\mifunction(\vec{\sigma})
, \sigma_e^{**},
\vec{x}_{(T \cup K)}''') + \epsilon & \mbox{[by
Proposition~\ref{prop:eps-robustness}]} \\ 
= &u_i(\Gamma_d,
\mifunction(\vec{\sigma})_{-T}
, \vec{\tau}_T,
\sigma'_e, \vec{x}_{(T \cup K)})& \end{array}$$
  which shows that 
  $\mifunction(\vec{\sigma} + \sigma_d)$
   is $\epsilon$-$(k,t)$-robust.
\commentout{
 The desired result follows from a straightforward application of Proposition~\ref{prop:adversary-immunity} and Corollary~\ref{prop:adversary-robustness}.
 }

It remains to compute the bound on the number of messages sent.
Clearly, if all we need is a weak implementation, then $n$ messages
suffice.  In the case of implementation, at most $2Rn$ messages are
sent (at most $R$ by each player, and at most $Rn$ by the mediator),
where $R$ is the least integer such that 
$(Rn)! \ge
|\mathcal{S}^{det}_{\Gamma_d,e}/\!\!\sim\!\!|$.  Thus, we must compute
an estimate for the number of equivalence of deterministic schedulers. 

A deterministic scheduler is a function from \emph{message patterns}
to a choice of message to be delivered, where a message pattern
describes which messages have been sent so far, which were delivered,
and the order that the messages were sent and delivered, but not the
contents of the message.  For example, taking $(s,i,j,k)$ (resp., 
$(d,i,j,k)$) to denote the $k$th message sent by player $i$ to player
$j$ (resp., that the $k$th message sent by player $i$ to $j$ is
received by $j$), and taking the mediator to be player 0, then a
typical message pattern might be $(s,0,3,1)$, $(s,1,0,1)$,
$(s,0,3,2)$, $(d,0,3,2)$ is the message pattern where the mediator
first send a message to player 3, then player 1 sends a message t the
mediator, then the mediator sends a second message to player 3, and
then then mediator's second message to player 3 is delivered.  Given this
message pattern, the scheduler can choose to deliver the mediator's
first message to player 3 or the message from player 1 to the
mediator.  Since the mediator sends at most $r$ messages to each
player, and each player sends at most $r$ messages to the mediator,
message patterns have length at most $4rn$.  The messages sent by a
player to the mediator are numbered consecutively, as are the messages
sent  by the mediator to each player.  A straightforward computation
then shows that there are at most $(4rn)!/(r!)^{2n}$ message patterns
of length $4rn$.  It is easy to see that there are fewer message
patterns of length $k$ for $k < 4rn$, so there are clearly at most
$(4rn)(4rn)!/(r!)^{2n}$ message patterns of length at most $4rn$.
A message pattern can have at most $2rn$ undelivered messages, so 
there are at most $(2rn)^{(4rn)(4rn)!/(r!)^{2n}}$ equivalence classes
of schedulers.
A straightforward application of Stirling's approximation formula ($n!
\sim (n/e)^n\sqrt{2\pi n}$) shows that $R = (4rn)^{4rn}$ suffices
for our purposes.
Since $rn$ is a bound on the number of messages sent, we have $N=rn$,
completing the proof.
\commentout{
  Thus, a
deterministic scheduler can be seen as a function from such lists to
the next message to be delivered. With this setup we proceed to bound
$\mathcal{S}^{det}_{\Gamma_d,e}/\!\!\sim$: 
Since each player sends at most $r$ messages, the mediator replies at
most $nr$ times, and thus the scheduler's local history consists of at
most $2nr$ messages. There are $(2nr)!$ ways in which these messages
arrive and are received relatively to each other, and therefore there
are at most $((2nr)!)^2$ possible local histories for the
scheduler. Since the scheduler has to pick a message for each of these
local histories, there are at most $2nr((2nr)!)^2$ elements in
$\mathcal{S}^{det}_{\Gamma_d,e}/\!\!\sim$. 

The number of messages in our construction is at most $(R+1)n$, in
which $R$ is the minimum number that satisfies $(Rn)! \ge
\left|\mathcal{S}^{det}_{\Gamma_d,e}/\!\!\sim\right|$. Clearly,
$(4rn)! \ge 2nr((2nr)!)^2$, which means that $R \le 4r$ and the number
of messages used in our previous construction is \msgbound. 
}
\end{proof}
}

Thus,
without loss of generality, we
can assume that the players and mediator use 
such 
a
strategy profile.
We call $\mifunction(\vec{\sigma} + \sigma_d)$ the \emph{minimally
  informative} strategy corresponding to $\vec{\sigma} +
\sigma_d$. More generally, 
  we say that $\vec{\sigma}^m + \sigma_d^m$ is a minimally informative
strategy if $\vec{\sigma}^m + \sigma_d^m = \mifunction(\vec{\sigma} +
\sigma_d)$ for some strategy profile $\vec{\sigma} + \sigma_d$. 

Since we consider only mediator games
in canonical form, this guarantees termination for all honest players
regardless of the strategy of rational and malicious players, provided
that the scheduler is standard (i.e., not relaxed).  
However, once we allow relaxed schedulers, there is a possibility of deadlock.
We assume for the purposes of the proof that we use the AH approach in
the mediator game, and have the players play the punishment punishment
strategy in their wills.
Since $\vec{\sigma}_{\ACT}$ guarantees $t$-cotermination 
for $t < n/3$, it follows that in the cheap-talk game, either all
honest players terminate or all honest players play the punishment
strategy.
This guarantees that the players get the same payoff in corresponding
histories  in  the mediator game and the cheap-talk game.

The next step in proving Theorem~\ref{thm:punish} is to show
  that rational players playing with a relaxed
scheduler cannot get an expected payoff
that is higher than their expected payoff when they play
such
a 
minimally informative
$(k,t)$-robust equilibrium strategy
with a non-relaxed scheduler.

\arxiv{
\begin{proposition}\label{lemma:better-mediator}
  If $\epsilon > 0$, $\vec{\sigma}+\sigma_d$ is a minimally informative
  (strongly) $\epsilon$-$(k,t)$-robust
equilibrium in a mediator game $\Gamma_d$
for which a $(k+t)$-punishment strategy exists,
$\sigma_E$ 
is a relaxed scheduler, $K$ and $T$ are disjoint sets of players
  with $1 \le |K| \le k$ and $|T| \le t$,
and  $\vec{\tau}_{(K \cup T)}$ is a strategy profile for the players in $K
  \cup T$, then
  there exists a value $\epsilon_0 < \epsilon$ such that
 \commentout{ 
 there exists a (non-relaxed) scheduler
 $\sigma_e$
  such that 
 for all input profiles $\vec{x}$
 }
 for all non-relaxed schedulers $\sigma_e$ and all input
 profiles $\vec{x}$,  we have that
\commentout{ 
 $$u_i'(\Gamma_d, (\vec{\sigma}_{(K \cup T)}, \vec{\tau}_{(K \cup T)}), \sigma_E
   \vec{x}_{(K \cup T)}) \le u_i'(\Gamma_d, (\vec{\sigma}_{-(K \cup T)},
  \vec{\tau}_{(K \cup T)}), \sigma_e, \vec{x})$$ for all $i \notin T$. 
  }
  $$u_i'(\Gamma_d, (\vec{\sigma}_{-(K \cup T)}, \vec{\tau}_{(K \cup T)}), \sigma_E, 
   \vec{x}_{(K \cup T)}) < u_i'(\Gamma_d, (\vec{\sigma}_{-T},
      \vec{\tau}_{T}), \sigma_e, \vec{x}_T) + \epsilon_0 $$ for some (for all) $i
   \notin T$.  
\end{proposition}

To prove Proposition~\ref{lemma:better-mediator}, we need a
preliminary lemma that characterizes deadlocks in mediator games in
canonical form with relaxed schedulers. It also shows that the
scheduler can detect when such a deadlock happens.

\begin{lemma}\label{prop:enhanced-scheduler-deadlock}
A run in a mediator game in canonical form with a relaxed scheduler $\sigma_E$
ends in deadlock iff
at some point in the run no player has received a STOP message and the
scheduler does not deliver any of the messages not yet
delivered. 
\end{lemma}

\begin{proof}
  Clearly, if no player has received a STOP messages and all the messages
  not yet delivered are never delivered, then the run ends in deadlock. To
show that all deadlocks must be of this form, assume that a
player receives eventually a STOP message. Then, since the mediator
sent all STOP messages at the same time, all other players are
guaranteed to receive a STOP message as well, given our assumption
that a relaxed scheduler
delivers either all or none of the messages sent at the same time.
\end{proof}

\begin{proof}[Proof of Proposition~\ref{lemma:better-mediator}]
We can view the adversary's strategy $(\vec{\tau}_{(K \cup T)}, \sigma_E)$ as a combination of (possibly infinitely many) deterministic strategies $(\vec{\tau}'_{(K \cup T)}, \sigma_E')$. Thus, it suffices to show the desired result for each of such deterministic strategies.

Let $(\vec{\tau}'_{(K \cup T)}, \sigma_E')$ be a deterministic
strategy for the adversary in the support of $(\vec{\tau}_{(K \cup
    T)}, \sigma_E)$. By the properties of our construction, the
fact that the adversary is deterministic, and the fact
that a relaxed scheduler must either deliver all the STOP messages 
from the mediator or deliver
none, it follows that, for a given input $\vec{x}_{(K\cup T)}$, either
all runs end in deadlock or all honest  
players terminate.

Suppose that for
some deterministic adversary $(\vec{\tau}'_{(K \cup T)}, \sigma_E')$
and input $\vec{x}_{K \cup T}$
all honest players terminate. Consider a non-relaxed
scheduler $\sigma_e'$  
that
acts just like $\sigma_E'$, except that  whenever it detects a deadlock (as
characterized by Lemma~\ref{prop:enhanced-scheduler-deadlock}, using
the fact that the mediator's $r$th message to each player includes
STOP), it instead delivers a message chosen at random. Then, $\sigma_E'$ and $\sigma_e'$ are indistinguishable when the players 
in $K \cup T$
play $\vec{\tau}'_{(K \cup T)}$. Therefore, there exists a value $\epsilon' < \epsilon$ such that
$$\begin{array}{lll}
& u_i'(\Gamma_d, (\vec{\sigma}_{-(K \cup T)}, \vec{\tau}'_{(K \cup T)}), \sigma_E', 
   \vec{x}_{(K \cup T)}) & \\
 = & u_i'(\Gamma_d, (\vec{\sigma}_{-(K \cup T)}, \vec{\tau}'_{(K \cup T)}), \sigma_e', 
   \vec{x}_{(K \cup T)}) &  \\
 < & u_i'(\Gamma_d,( \vec{\sigma}_{-T}, \vec{\tau}_T), \sigma_e,
 \vec{x}_T) +
\epsilon'
   & \mbox{[by
Proposition~\ref{prop:compactness-robustness}]}\end{array} $$
for some (for all) $i \in T$ and all non-relaxed schedulers
$\sigma_e$. Suppose instead that all runs with adversary
$(\vec{\tau}'_{(K \cup T)}, \sigma_E')$ end in deadlock. Then, there
exists a value $\epsilon' < \epsilon$ such that  
%
$$\begin{array}{lll} & u_i'(\Gamma_d, (\vec{\sigma}_{-(K \cup T)}, \vec{\tau}_{(K \cup T)}), \sigma_E',  
   \vec{x}_{(K \cup T)}) & \\
    < & u_i'(\Gamma_d, \vec{\sigma}, \sigma_e, \vec{x}_{(K \cup T)}) & \mbox{[by definition of $(t+k)$-punishment strategy}] \\
    < & u_i'(\Gamma_d, (\vec{\sigma}_{-T}, \vec{\tau}_{T}), \sigma_e,
    \vec{x}_T) + 
    \epsilon'
     & \mbox{[by
        Proposition~\ref{prop:compactness-immunity}]}\end{array} $$

Using a compactness argument analogous to that of
Proposition~\ref{prop:compactness-immunity}, there exists a value
$\epsilon_0 < \epsilon$ such that $\epsilon' \le \epsilon_0$ for all
relaxed adversaries $(\vec{\tau}'_{(K \cup T)}, \sigma_E')$ in the support of $(\vec{\tau}_{(K \cup T)}, \sigma_E)$ and all
inputs $\vec{x}_{(K \cup T)}$. So
%
$$u_i'(\Gamma_d, (\vec{\sigma}_{-(K \cup T)}, \vec{\tau}_{(K \cup T)}), \sigma_E, 
   \vec{x}_{(K \cup T)}) < u_i'(\Gamma_d, (\vec{\sigma}_{-T},
   \vec{\tau}_{T}), \sigma_e, \vec{x}_T) + \epsilon_0$$ for all
   inputs $\vec{x}_{(K \cup T)}$,
   all non-relaxed schedulers $\sigma_e$, all
   strategies $\vec{\tau}_{(T \cup K)}$, and for some (for all) $i 
   \notin T$.
\end{proof}
\commentout{

\commentout{
 and the fact that a relaxed
scheduler must either deliver all the STOP messages
from the mediator or deliver none, it follows that
we can assume without loss of generality (using the same argument as
in the proof of Lemma~\ref{lem:correctness-of-construction}) that
rational and malicious players play strategy profile $\vec{\sigma}$.
}
Using the same argument as in the proof of Lemma~\ref{lem:correctness-of-construction}, for every strategy profile $\vec{\tau}_{(K \cup T)}$ for players in $K \cup T$ and input profile $\vec{x}_{(K \cup T)}$, there exists an input profile $\vec{x}'_{(K \cup T)}$ and a relaxed scheduler $\sigma_E'$ such that 
$$u_i'(\Gamma_d, \vec{\sigma}_{-(K \cup T)}, \vec{\tau}_{(K \cup T)}, \sigma_E, 
   \vec{x}_{(K \cup T)}) = u_i'(\Gamma_d, \vec{\sigma}, \sigma_E', 
   \vec{x}'_{(K \cup T)})$$   
Thus, it suffices to show that
$$u_i'(\Gamma_d, \vec{\sigma}, \sigma_E, 
   \vec{x}_{(K \cup T)}) < u_i'(\Gamma_d, (\vec{\sigma}_{-T},
      \vec{\tau}_{T}), \sigma_e, \vec{x}'_T) + \epsilon_0$$
      for all non-relaxed schedulers $\sigma_e$ and all input profiles $\vec{x}_{T}$ of players in $T$.

    We can view $\sigma_E$ as a
probability distribution on (possibly infinitely many) deterministic relaxed
schedulers $\sigma_E'$. 
 For each relaxed scheduler $\sigma_E'$ in the support of $\sigma_E$,
 consider a non-relaxed 
 scheduler $\sigma_e'$ that  
  acts just like $\sigma_E'$, except that  whenever it detects a deadlock (as
characterized by Lemma~\ref{prop:enhanced-scheduler-deadlock}, using
the fact that the mediator's $r$th message to each player includes
STOP), it instead delivers a message chosen at random.

\commentout{
Because of the properties of our construction 
 and the fact that a relaxed
scheduler must either deliver all the STOP messages
from the mediator or deliver none, it follows that
we can assume without loss of generality (using the same argument as in the proof of Lemma~\ref{lem:correctness-of-construction}) that rational and malicious players play strategy profile $\vec{\sigma}$. 
}
Because of the properties of our construction 
and the fact that a relaxed
scheduler must either deliver all the STOP messages
from the mediator or deliver none it follows that 
for a given input $\vec{x}_{(K\cup T)}$, with a 
deterministic relaxed scheduler, either all runs end in deadlock or all honest
players terminate with probability 1.
%
Suppose that for
some deterministic scheduler $\sigma_E'$ in the support of $\sigma_E$
all runs terminate with
probability 1. Then, $\sigma_E'$ and $\sigma_e'$ are
indistinguishable when the players 
in $K \cup T$
play 
$\vec{\sigma}_{(K \cup T)}$
with input $\vec{x}_{(K \cup T)}$.  Thus, there
exists a value $\epsilon' < \epsilon$ such that
%
$$\begin{array}{lll}
& u_i'(\Gamma_d, \vec{\sigma}, \sigma_E', 
   \vec{x}_{(K \cup T)}) & \\
 = & u_i'(\Gamma_d, \vec{\sigma}, \sigma_e',
   \vec{x}_{(K \cup T)}) &  \\
 < & u_i'(\Gamma_d, \vec{\sigma}, \sigma_e,
 \vec{x}_T) +
\epsilon'
   & \mbox{[by
     Proposition~\ref{prop:compactness-robustness}]}\end{array} $$ 
%
 for some (for all) $i
   \notin T$
   and all non-relaxed schedulers $\sigma_e$.  
   Suppose instead that all runs with scheduler
  $\sigma_E'$ end in deadlock. Then there exists a value $\epsilon'   
   < 
\epsilon$ such that
 $$\begin{array}{lll} & u_i'(\Gamma_d, \vec{\sigma}, \sigma_E',  
   \vec{x}_{(K \cup T)}) & \\
    < & u_i'(\Gamma_d, \vec{\sigma}, \sigma_e, \vec{x}_{(K \cup T)}) & \mbox{[by definition of $(t+k)$-punishment strategy}] \\
    < & u_i'(\Gamma_d, \vec{\sigma}, \sigma_e,
    \vec{x}_T) + 
    \epsilon'
     & \mbox{[by
        Proposition~\ref{prop:compactness-immunity}]}\end{array} $$
        for some (for all) $i 
   \notin T$
   and all non-relaxed schedulers $\sigma_e$.      

\commentout{
      Using an argument analogous to that of
Proposition~\ref{prop:compactness-immunity}, we can pick these 
$\epsilon'$
values such that they hold for all inputs $\vec{x}_{(K
    \cup T)}$. Note that $\sup_{\sigma_e' \mbox{ in support
    of $\sigma_E$}}\epsilon' 
    < 
    \epsilon$, so
%
}
Using a compactness argument analogous to that of
Proposition~\ref{prop:compactness-immunity}, there exists a value
$\epsilon_0 < \epsilon$ such that $\epsilon' \le \epsilon_0$ for all
relaxed schedulers $\sigma_E'$ in the support of $\sigma_E$ and all
inputs $\vec{x}_{(K \cup T)}$. So
\commentout{
$$u_i'(\Gamma_d, (\vec{\sigma}_{-(K \cup T)}, \vec{\tau}_{(K \cup T)}), \sigma_E, 
   \vec{x}_{(K \cup T)}) < u_i'(\Gamma_d, (\vec{\sigma}_{-T},
   \vec{\tau}_{T}), \sigma_e, \vec{x}_T) + \epsilon'$$ for all
   inputs $\vec{x}_{(K \cup T)}$,
   all non-relaxed schedulers $\sigma_e$, all
   strategies $\vec{\tau}_{(T \cup K)}$, and for some (for all) $i 
   \notin T$ 
   }
   $$u_i'(\Gamma_d, \vec{\sigma}, \sigma_E, 
   \vec{x}_{(K \cup T)}) < u_i'(\Gamma_d, (\vec{\sigma}_{-T},
   \vec{\tau}_{T}), \sigma_e, \vec{x}_T) + \epsilon'$$
   }

An analogous argument can be used for the errorless case:
}

\begin{proposition}\label{lemma:better-mediator-errorless}
  If $\vec{\sigma}+\sigma_d$ is a minimally informative
  $(k,t)$-robust
    strategy in a mediator game $\Gamma_d(u_i')$
for which a $(k+t)$-punishment strategy exists,
$\sigma_E$ 
  is a relaxed scheduler, $T$ and $K$ are disjoint sets of players
  with $|T| \le t$ and $1 \le |K| \le k$,
and  $\vec{\tau}_{(K \cup T)}$ is a strategy profile for the players in $K
  \cup T$, then
 there exists a (non-relaxed) scheduler
 $\sigma_e$
  such that 
for all input profiles $\vec{x}$ and all $i \notin T$,
$$u_i'(\Gamma_d, (\vec{\sigma}_{-(K \cup T)}, \vec{\tau}_{(K \cup
  T)}), \sigma_d, \sigma_E, \vec{x}_{(K \cup T)}) \le u_i'(\Gamma_d, (\vec{\sigma}_{-T},
      \vec{\tau}_{T}), \sigma_d, \sigma_e, \vec{x}_T).$$ 
\end{proposition}
\disc{Intuitively, since no information is leaked, the expected outcome for
rational players conditional on their history does not change before
receiving a STOP message; thus, they have no incentive to deadlock the
computation.  A complete proof is given in the full paper.}

We can now prove 
(strong)
$(k,t)$-robustness. Let
$\Gamma_d(\vec{u}')$ be any utility variant of $\Gamma_d$ such that
$\vec{\sigma}+\vec{\sigma}_d$ is a $(k,t)$-robust equilibrium of $\Gamma_d(\vec{u}')$, let 
$\sigma'_e$ be any scheduler in $\Gamma_{\ACT}(\vec{u})$, and let $K$ and
$T$ be disjoint subsets of players with $1 \le |K| \le k$ and $|T| \le t$,
respectively, 
such that $3k+4t < n$. Let $\vec{\tau}_K$ and $\vec{\tau}_T$
be strategy profiles for players in $K$ and $T$, respectively.
By 
Theorem~\ref{thm:errorless-simulation}, 
there exists a function $H$ and
a relaxed scheduler $\sigma_E$ in the cheap-talk game such
that $$u_i'(\Gamma_{\ACT}(\vec{u}'), (\vec{\sigma}_{\ACT})_{-(K \cup T)},
\vec{\tau}_K, \vec{\tau}_T), \sigma'_e, \vec{x}_{(K \cup T)}) = u_i'(\Gamma_{d}(\vec{u}'),
(\vec{\sigma}_{-(K \cup T)}, H(\vec{\tau}_K), H(\vec{\tau}_T)), \sigma_d,
\sigma_E, \vec{x}_{(K \cup T)})$$ for 
some (for all)
 $i \notin T$ and all input profiles $\vec{x}$. 
We can assume without loss of generality that $(\vec{\sigma},
\sigma_d)$ is minimally informative. Thus, 
by Proposition~\ref{lemma:better-mediator-errorless}, there exists a
non-relaxed 
scheduler $\sigma_e$ such that
$$u_i'(\Gamma_{d}(\vec{u}'), (\vec{\sigma}_{-(K \cup T)},
H(\vec{\tau}_T), H(\vec{\tau}_K)),\sigma_d, \sigma_E, \vec{x}_{(K \cup T)}) \le
u_i'(\Gamma_{d}(\vec{u}'), (\vec{\sigma}_{-T}, H(\vec{\tau}_T)),
\sigma_d,\sigma_e, \vec{x}_{(K \cup T)}).$$ 
Finally, by Theorem~\ref{thm:errorless-simulation}, there exists a non-relaxed
scheduler $\sigma_e''$ such that $$\begin{array}{lll} &
  u_i'(\Gamma_{\ACT}(\vec{u}'), ((\vec{\sigma}_{\ACT})_{-T}, \vec{\tau}_T),
  \sigma'_e, \vec{x}_T) & \\ 
= &  u_i'(\Gamma_{d}(\vec{u}'), (\vec{\sigma}_{-T}, H(\vec{\tau}_T)),
\sigma_d,
\sigma_e'', \vec{x}_T) & \mbox{[by Theorem~\ref{thm:errorless-simulation}]} \\ 
= & u_i'(\Gamma_{d}(\vec{u}'), (\vec{\sigma}_{-T}, H(\vec{\tau}_T)),
\sigma_d,\sigma_e, \vec{x}_T) & 
\mbox{[by Corollary~\ref{lemma:scheduler-proof}].} \end{array}$$ 
Therefore, $$u_i'(\Gamma_{\ACT}(\vec{u}'),
((\vec{\sigma}_{\ACT})_{-(K \cup T)}, \vec{\tau}_K, \vec{\tau}_T),
\sigma'_e, \vec{x}_{(K \cup T)}) \le u_i'(\Gamma_{\ACT}(\vec{u}'), ((\vec{\sigma}_{\ACT})_{-T},
\vec{\tau}_T), \sigma'_e, \vec{x}_T),$$ as desired. 
\wbox

We remark that, with a little more effort, we can show that the
minimally informative strategy $f(\vec{\sigma}+\sigma_d)$ that
implements $\vec{\sigma}+\sigma_d$ is actually a $t$-bisimulation and
a $t$-emulation of of $\vec{\sigma}+\sigma_d$.  Moreover, the cheap-talk
strategy that implements $f(\vec{\sigma}+\sigma_d)$ preserves these
properties.  Thus, under the conditions of Theorem~\ref{thm:punish},
we can get a cheap-talk strategy that $t$-bisimulates and $t$-emulates
a strategy in the mediator game.

\disc{To prove Theorem~\ref{thm:punish-eps}, we use an argument analogous
to that used to prove Theorem~\ref{thm:punish}, using
Theorem~\ref{thm:error-simulation} instead of
Theorem~\ref{thm:errorless-simulation}.
We leave details to the full paper.}

\arxiv{
\subsection{Proof of Theorem~\ref{thm:punish-eps}}

To prove Theorem~\ref{thm:punish-eps}, we use an analogous strategy 
to that used to prove Theorem~\ref{thm:punish}, using
Theorem~\ref{thm:error-simulation} instead of
Theorem~\ref{thm:errorless-simulation}. The same argument as that used in
the proof  Theorem~\ref{thm:upperbound-default-no-punish1} shows that
for all $\epsilon' \in (0,1]$ there exists a protocol
  $\vec{\sigma}_{\ACT}$ that $\epsilon'$-implements $\vec{\sigma}+\sigma_d$
   and that $\vec{\sigma}_{\ACT}$
  is $(\epsilon, t)$-immune. 

    To prove 
        (strong) 
    $\epsilon$-$(k,t)$-robustness, fix an adversary $A =
(\vec{\tau}_K, \vec{\tau}_T, \sigma_e)$ for subsets $K,T$ such that
$1 \le |K| \le k$, $|T| \le t$ and $K \cap T = \emptyset$. By
Theorem~\ref{thm:error-simulation}, there exists a function $H$ from
    strategies to strategies and a relaxed scheduler $\sigma_E$ such
        that, for all input profiles $\vec{x}$, $$u_i(\Gamma_{\ACT},
((\vec{\sigma}_{\ACT})_{-(K \cup T)}, \vec{\tau}_K, 
\vec{\tau}_T), \sigma_e, \vec{x}_{(K \cup T)}) < u_i(\Gamma_d, (\vec{\sigma}_{-(K \cup T)},
%
H(\vec{\tau}_K), H(\vec{\tau}_T)),  \sigma_d, \sigma_E,\vec{x}_{(K \cup T)}) +
\epsilon'M$$  
for some (resp. for all)
$i \in K$.

By Theorem~\ref{thm:error-simulation}, there exists a non-relaxed
scheduler $\sigma'_e$ such that  
$$u_i(\Gamma_{\ACT}, ((\vec{\sigma}_{\ACT})_{-T}, \vec{\tau}_T), \sigma_e, \vec{x}_T)
> u_i(\Gamma_d, (\vec{\sigma}_{-T}, H(\vec{\tau}_T)), \sigma_d, \sigma'_e,
\vec{x}_T) - \epsilon'M.$$ 
Thus, 
we have that
$$\begin{array}{lll}
& u_i(\Gamma_{\ACT}, ((\vec{\sigma}_{\ACT})_{-(K \cup T)}, \vec{\tau}_K, \vec{\tau}_T), \sigma_e, \vec{x}_{(K \cup T)}) & \\
    < & u_i(\Gamma_d, (\vec{\sigma}_{-T}, H(\vec{\tau}_T)), \sigma_d, \sigma_e',
\vec{x}_T) + \epsilon_0 + \epsilon'M&   \mbox{[by
    Proposition~\ref{lemma:better-mediator}]}\\ 
< & u_i(\Gamma_{\ACT}, ((\vec{\sigma}_{\ACT}){-T}, \vec{\tau}_T), \sigma_e, \vec{x}_T) +
\epsilon_0 + 2\epsilon'M. & 
\end{array}$$
for some $\epsilon_0 < \epsilon$. Therefore, 
taking $\epsilon' := (\epsilon - \epsilon_0)/2M$,
 we have
that $\vec{\sigma}_{\ACT}$ is a 
(strongly)
$\epsilon$-$(k,t)$-robust equilibrium
for $\Gamma_{\ACT}$. 

Again, as was the case for Theorem~\ref{thm:punish}, with a little
more effort we can show that under the conditions of
Theorem~\ref{thm:punish-eps}, we can get a cheap-talk strategy that
$t$-bisimulates and $t$-emulates a strategy in the mediator game.
}


\bibliography{game1,joe}





\end{document}